\documentclass[10pt,twoside,reqno,a4paper]{article}
\usepackage[T1]{fontenc}
\usepackage{a4wide}

\usepackage{amsfonts, amsmath, amssymb,latexsym,nccmath}
\usepackage[activeacute,english]{babel}
\usepackage{epsfig}
\usepackage{graphicx}
\usepackage[font=sf,sc]{caption}
\usepackage{float}
\usepackage{cancel}
\usepackage{color}
\usepackage[utf8]{inputenc}
\usepackage{amsthm}
\usepackage{bm}
\usepackage[all]{xy}
\usepackage{enumerate}
\usepackage{comment}
\usepackage[colorinlistoftodos]{todonotes} 
\usepackage{bbm}

\usepackage[colorlinks=true,linkcolor=blue,urlcolor=blue,citecolor=blue]{hyperref}

\newcommand{\R}{{\mathbb {R}}}

\renewcommand{\leq}{\leqslant}
\renewcommand{\geq}{\geqslant}

\newcommand{\RR}{\mathbb{R}}

\newtheorem{theorem}{Theorem}
\newtheorem{thm}[theorem]{Theorem}
\theoremstyle{plain}

\newtheorem{definition}[theorem]{Definition}

\newtheorem{lem}[theorem]{Lemma}

\newtheorem{prop}[theorem]{Proposition}
\newtheorem{remark}[theorem]{Remark}
\newtheorem{obs}[theorem]{Observation}

\numberwithin{equation}{section}
\numberwithin{theorem}{section}

\allowdisplaybreaks

\theoremstyle{definition}

\usepackage{enumerate}

\usepackage[
backend=bibtex,
style=numeric,
sorting=nyt,
giveninits=true, 
doi=false,isbn=false,url=false,eprint=true,sortcites=true,maxbibnames=99,
]{biblatex}
\DeclareNameAlias{author}{family-given}
\appto{\bibsetup}{\sloppy}
\newtheorem{manuallemmainner}{Lemma}
\newenvironment{manuallemma}[1]{%
  \renewcommand\themanuallemmainner{#1}%
  \manuallemmainner
}{\endmanuallemmainner}

\newtheorem{manualpropinner}{Proposition}
\newenvironment{manualprop}[1]{%
  \renewcommand\themanualpropinner{#1}%
  \manualpropinner
}{\endmanualpropinner}

\usepackage{csquotes}
\RequirePackage{dsfont} 
\newcommand{\EE}{\mathbb{E}} 
\newcommand{\PP}{\mathbb{P}} 
\newcommand{\QQ}{\mathbb{Q}} 
\newtheorem{ass}{Assumption}

\title{Thiele's PIDE for unit-linked policies in the Heston-Hawkes stochastic volatility model}

\author{David R. Ba\~{n}os\thanks{Department of Mathematics, University of Oslo, P.O. Box 1053 Blindern, N-0316 Oslo, Norway, Email: davidru@math.uio.no}, 
Salvador Ortiz-Latorre\thanks{Department of Mathematics, University of Oslo, P.O. Box 1053 Blindern, N-0316 Oslo, Norway, Email: salvadoo@math.uio.no}, 
and 
Oriol Zamora Font\thanks{Department of Mathematics, University of Oslo, P.O. Box 1053 Blindern, N-0316 Oslo, Norway, Email: oriolz@math.uio.no}}
\date{\today}

\begin{document}
\maketitle

\begin{abstract}
The main purpose of the paper is to derive Thiele's differential equation for unit-linked policies in the Heston-Hawkes stochastic volatility model introduced in \cite{arxiv}. This model is an extension of the well-known Heston model that incorporates the volatility clustering feature by adding a compound Hawkes process in the volatility. Since the model is arbitrage-free, pricing unit-linked policies via the equivalence principle under a risk neutral probability measure is possible. Studying the moments of the variance and certain stochastic exponentials, a suitable family of risk neutral probability measures is found. The established and practical method to compute reserves in life insurance is by solving Thiele's equation, which is crucial to guarantee the solvency of the insurance company. 


{\it AMS classification MSC2020}: 60G55, 60H30, 91G05, 91G15. 

{\it Declarations of interest}: none.
\end{abstract}

\section{Introduction}\label{Intro}

A life insurance policy is a contract that defines the evolution of the cash flow between the insurance company and the insured. In traditional life insurance contracts, the benefit provided to the insured if the pre-agreed conditions are met is fixed. In contrast, the distinctive property of unit-linked life insurance policies is that the benefit is linked to the market value of some fixed portfolio composed of stocks, bonds or other financial assets. Consequently, two types of risk prevail in unit-linked policies, the mortality risk that controls the future flow of payments between the two parties and the financial risk arising from the future returns on the financial investments.

The computation of the so-called reserve is the basis to guarantee the solvency of the insurance company against the two aforementioned risks. Namely, the reserve is the present value of future potential payments from the insurance company to the insured and is the ordinary procedure to determine the premiums. The established method to calculate the reserve is by solving the celebrated Thiele's differential equation, which dates back to 1875. It was first published in the obituary on Thiele in \cite{gram1910professor} and presented in the scientific paper \cite{germanpaper} in 1913.

Literature on pricing unit-linked policies is vast and we mention here a non-exhaustive list. Brennan and Schwartz \cite{BRENNAN1976195,f7a4f9e4-b3a8-3df5-afdf-dc140a564965} and Boyle and Schwartz \cite{eb7c44a5-55b7-3c26-944a-604759db8bcf} conducted, to the best of our knowledge, the first studies of unit-linked contracts using modern financial techniques. Delbaen \cite{delbaen} and Bacinello and Ortu \cite{bacinello} examined unit-linked contracts employing the martingale-based theory developed by Harrison and Kreps \cite{Harrison}. See also \cite{svein} where Aase and Persson revisit the principle of equivalence under a risk neutral probability measure for the pricing of unit-linked policies after presenting some historical context on the topic. 

It is acknowledged in life insurance that in the case of long-term contracts interest rate risk plays an important role. To incorporate that risk into the valuation of life insurance policies, Norberg and Møller obtained Thiele's differential equation assuming stochastic interest rates in \cite{norberg} and, later on, Persson derived the risk adjusted version of it in \cite{srate6}. For more fundamental references about interest rate risk in unit-linked policies see \cite{Persson1994PricingLI,srate1,srate3,srate5} and some further extensions \cite{baños2020life,Ahmad_Bladt_2023,srate2,srate4}.

It is worth pointing out that also policyholder behaviour has a significant impact in life insurance. For instance, the insured may surrender the contract, cancel all future cash flows and receive just a single payment. Similarly, the policyholder may cancel the future premiums and accept a reduction in the benefits. See \cite{beh1,beh2,beh3,beh4,beh5,beh7,beh6,beh8} for references where behavioural risk in life insurance policies is studied.

The scope of this paper is to study the financial risk in unit-linked policies coming from stock returns. Precisely, we derive Thiele's differential equation under the assumption that the benefit is linked to a stock that follows the Heston-Hawkes stochastic volatility model presented in \cite{arxiv}. This model is an extension of the well-known Heston model that incorporates the volatility clustering feature by adding a compound Hawkes process component in the volatility process. A Hawkes process is a self-exciting point process introduced in 1917 by Hawkes, see \cite{hawkes1,hawkes2}, with many applications in high-frequency finance, insurance, seismology and other disciplines. By exploiting the tractability of the model, it is shown in \cite{arxiv} that the Heston-Hawkes model is arbitrage-free and incomplete. Furthermore, the passage from the historical probability to the risk neutral ones is made explicit via a family of equivalent martingale measures. 

There is an extensive literature supporting the inclusion of jumps in the volatility, for instance, see \cite{Pan02,eraker,Erakeretal03,Cont07,Bates00,bates2}. In particular, the authors in \cite[Section 5.4]{Pan02} provide an empirical study on the higher moments of the volatility process searching for evidence of jumps in volatility, as it was conjectured in \cite{Bates00}. 

It is a common practice in actuarial science to consider two independent probability spaces, one to describe the states of the insured and the other one to describe the evolution of the financial investments. As explained in \cite{svein}, the right methodology to price unit-linked policies is through the principle of equivalence under $\QQ$, a risk neutral probability measure in the financial probability space. Therefore, the arbitrage-free property of the stochastic volatility model proven in \cite{arxiv} is required to make sense of the pricing of unit-linked policies.

The paper is organized as follows. In Section \ref{pre} we present the Heston-Hawkes stochastic volatility model introduced in \cite{arxiv}. In addition, we give some preliminary lemmas required to derive Thiele's differential equation. Specifically, we provide several technical results regarding the compensator of the Hawkes process under the risk neutral probability measures. The proofs are postponed in the Appendix for the sake of clarity. Finally, in Section \ref{main} we obtain Thiele's differential equation for unit-linked policies under this stochastic volatility model.

\section{Preliminary results}\label{pre}
The objective of this paper is to derive Thiele's differential equation for unit-linked policies under the Heston-Hawkes stochastic volatility model. First, we present the Heston-Hawkes stochastic volatility model introduced in \cite{arxiv}. Then, we summarize results proven in \cite{arxiv} regarding the arbitrage-free property of the model, which is necessary for the pricing of unit-linked policies. Additionally, we investigate the existence of (positive and negative) moments of the variance and (positive) moments of the Radon–Nikodym derivative of the change of measure. The existence of those moments will be used to prove that the compensated Hawkes process is a martingale under a suitable family of risk neutral probability measures. For the sake of clarity, all the proofs of this section are postponed to the Appendix. 

\subsection{Stochastic volatility model}
We outline the Heston-Hawkes stochastic volatility model introduced in \cite{arxiv}. This model is an extension of the Heston model that incorporates the volatility clustering effect by adding an independent compound Hawkes process to the variance. 

Let $T>0$ be a fixed time horizon. On a complete probability space $(\Omega,\mathcal{A},\PP)$, we consider a two-dimensional standard Brownian motion $(B,W)=\left\{\left(B_t,W_t\right), t\in[0,T]\right\}$ and its minimally augmented filtration $\mathcal{F}^{(B,W)}=\{\mathcal{F}^{(B,W)}_t, t\in[0,T]\}$. On $(\Omega,\mathcal{A},\PP)$, we also consider a Hawkes process $N=\{N_t, t\in[0,T]\}$, this is a counting process with stochastic intensity given by
\begin{align*}
    \lambda_t=\lambda_0+\alpha\int_0^t e^{-\beta(t-s)}dN_s,
\end{align*}
or, equivalently, 
\begin{align}\label{dynlambda}
    d\lambda_t=-\beta(\lambda_t-\lambda_0)dt+\alpha dN_t,
\end{align}
where $\lambda_0>0$ is the initial intensity, $\beta>0$ is the speed of mean reversion and $\alpha\in(0,\beta)$ is the self-exciting factor. Note that the stability condition holds because $\alpha<\beta$. For further details on Hawkes processes see \cite[Section 2]{Hawkesfinance} and \cite[Section 3.1.1]{article1}. 
Then, we consider a sequence of i.i.d., strictly positive and integrable random variables $\{J_i\}_{i\geq 1}$ and the compound Hawkes process $L=\{L_t, t\in[0,T]\}$ given by
\begin{align*}
    L_t=\sum_{i=1}^{N_t}J_i.
\end{align*}  
We make an assumption on the moment generating function $M_J(t)=\EE[\exp(tJ_1)]$ of $J_1$. Note that since $J_1$ is strictly positive, $M_J$ is well defined at least on the interval $(-\infty,0]$.
\begin{ass}\label{as}
   There exists $\epsilon_J>0$ such that $M_J$ is well defined in $(-\infty,\epsilon_J)$. Moreover, $(-\infty,\epsilon_J)$ is the maximal domain in the sense that
\begin{align*}
    \lim_{t\rightarrow\epsilon_J^-}M_J(t)=\infty.
\end{align*}
Since $\epsilon_J>0$, all positive moments of $J_1$ are finite.
\end{ass}
 We assume that $(B,W), N$ and $\{J_i\}_{i\geq 1}$ are independent of each other. We write $\mathcal{F}^L=\left\{\mathcal{F}^L_t, t\in[0,T]\right\}$ for the minimally augmented filtration generated by $L$ and
\begin{align*}
    \mathcal{F}=\{\mathcal{F}_t=\mathcal{F}_t^{(B,W)}\vee\mathcal{F}_t^L, t\in[0,T]\},
\end{align*}
for the joint filtration. We assume that $\mathcal{A}=\mathcal{F}_T$ and we will work with $\mathcal{F}$. Since $(B,W)$ and $L$ are independent processes, $(B,W)$ is also a two-dimensional $(\mathcal{F},\PP)$-Brownian motion. 

We assume that the interest rate is deterministic and constant equal to $r$. Finally, with all these ingredients, we introduce the Heston-Hawkes model. The stock price $S=\{S_t, t\in[0,T]\}$ and its variance $v=\{v_t, t\in[0,T]\}$ are given by
\begin{align}\label{2}
    \frac{dS_t}{S_t} & =\mu_tdt+\sqrt{v_t}\left(\sqrt{1-\rho^2} dB_t+\rho dW_t\right), \\
    \label{21} dv_t & =-\kappa\left(v_t-\Bar{v}\right)dt+\sigma\sqrt{v_t}dW_t+\eta dL_t,
\end{align}
where $S_0>0$ is the initial price of the stock, $\mu: [0,T] \rightarrow \R$ is a measurable and bounded function, $\rho\in(-1,1)$ is the correlation factor, $v_0>0$ is the initial value of the variance, $\kappa>0$ is the variance's mean reversion speed, $\Bar{v}>0$ is the long-term variance, $\sigma>0$ is the volatility of the variance and $\eta>0$ is a scaling factor. We assume that the Feller condition $2\kappa\Bar{v}\geq\sigma^2$ is satisfied, see \cite[Proposition 1.2.15]{AlfonsiAurélien2015ADaR}. For more details and results on the model see \cite{arxiv}.

\subsection{Arbitrage-free property}
We recall the arbitrage-free property of the Heston-Hawkes stochastic volatility model, as proved in \cite{arxiv}. The main result is Theorem \ref{risk}, and prior to that, we give results about the existence and positivity of the variance process as well as the expectation of the exponential of the integrated variance. As described in \cite{svein}, the right procedure to price unit-linked policies is through the principle of equivalence under  a risk neutral probability measure. Therefore, the arbitrage-free property of the stochastic volatility model proven in \cite{arxiv} is required to derive Thiele's differential equation. 

\begin{prop}\label{p2}
Equation \eqref{21} has a pathwise unique strong solution. 
\end{prop}
\begin{proof}
    See \cite[Proposition 2.1]{arxiv}.
\end{proof}

\begin{prop}\label{p1}
Let $\widetilde{v}=\{\widetilde{v}_t, t\in[0,T]\}$ be the pathwise unique strong solution of 
\begin{align}\label{sde10}
    \widetilde{v}_t=v_0-\kappa\int_0^t\left(\widetilde{v}_s-\Bar{v}\right)ds+\sigma\int_0^t\sqrt{\widetilde{v}_s}dW_s.
\end{align}
Then,
\begin{align*}
    \PP\left(\left\{\omega\in\Omega : \widetilde{v}_t(\omega)\leq v_t(\omega) \ \forall t\in[0,T]\right\}\right)=1,
\end{align*}
where $v$ is the pathwise unique strong solution of \eqref{21}. In particular, $v$ is a strictly positive process.
\end{prop}
\begin{proof}
    See \cite[Proposition 2.2]{arxiv}.
\end{proof}
\begin{prop}\label{novikov}
For $c\leq\frac{\kappa^2}{2\sigma^2}$, define $D(c):=\sqrt{\kappa^2-2\sigma^2c}$, $\Lambda(c):=\frac{2\eta c\left(e^{D(c)T}-1\right)}{D(c)-\kappa+\left(D(c)+\kappa\right)e^{D(c)T}}$ and
\begin{align*}
    c_l:=\sup\left\{c\leq\frac{\kappa^2}{2\sigma^2}: \Lambda(c)< \epsilon_J \hspace{0.3cm}\text{and}\hspace{0.3cm} M_J\left(\Lambda(c)\right)\leq\frac{\beta}{\alpha}\exp\left(\frac{\alpha}{\beta}-1\right) \right\}.
\end{align*}
Then, $c_l>0$ and for $c<c_l$
\begin{align*}
    \EE\left[\exp\left(c\int_0^Tv_udu\right)\right]<\infty.
\end{align*}
\end{prop}
\begin{proof}
    See \cite[Lemma 3.1]{arxiv} and \cite[Proposition 3.1]{arxiv}.
\end{proof}
\begin{thm}\label{risk}Let $a\in\RR$ and define $\theta_t^{(a)}:=\frac{1}{\sqrt{1-\rho^2}}\left(\frac{\mu_t-r}{\sqrt{v_t}}-a\rho\sqrt{v_t}\right)$, 
\begin{align*}
    Y_t^{(a)} &:=\exp\left(-\int_0^t\theta_u^{(a)}dB_u-\frac{1}{2}\int_0^t(\theta_u^{(a)})^2du\right), \\
    Z_t^{(a)} &:=\exp\left(-a\int_0^t\sqrt{v_u}dW_u-\frac{1}{2}a^2\int_0^tv_udu\right)
\end{align*}
and $X_t^{(a)}:=Y_t^{(a)}Z_t^{(a)}$. Recall the definition of $c_l$ in Proposition \ref{novikov}. \begin{enumerate}
    \item $X^{(a)}$ is a $(\mathcal{F},\PP)$-martingale for $|a|<\sqrt{2c_l}$.
    \item The set 
\begin{align}\label{elmmset}
    \mathcal{E}:=\left\{\QQ(a) \hspace{0.2cm}\text{given by}\hspace{0.2cm} \frac{d\QQ(a)}{d\PP}=X_T^{(a)} \hspace{0.2cm}\text{with}\hspace{0.2cm} |a|<\sqrt{2c_l}\right\} 
\end{align}
is a set of equivalent local martingale measures.
\item Let $\QQ(a)\in\mathcal{E}$, the process $(B^{\QQ(a)},W^{\QQ(a)})$ defined by
    \begin{align}\label{BMQ}
    dB_t^{\QQ(a)}&:=dB_t+\theta_t^{(a)}dt, \notag\\
    dW_t^{\QQ(a)}&:=dW_t+a\sqrt{v_t}dt
\end{align}
is a a two-dimensional standard $(\mathcal{F},\QQ(a))$-Brownian motion.
\item Let $\QQ(a)\in\mathcal{E}$, the dynamics of $S$ and $v$ are given by
\begin{align}\label{stockqa}
    \frac{dS_t}{S_t} &=rdt+\sqrt{v_t}\left(\sqrt{1-\rho^2} dB_t^{\QQ(a)}+\rho dW_t^{\QQ(a)}\right), \\
     \label{volqa}dv_t &= -\kappa^{(a)}(v_t-\Bar{v}^{(a)})dt+\sigma\sqrt{v_t}dW_t^{\QQ(a)}+\eta dL_t.
\end{align}
where $\kappa^{(a)}=\kappa+a\sigma$ and $\Bar{v}^{(a)}=\frac{k\Bar{v}}{k+a\sigma}$.
\item If $\rho^2<c_l$, the set 
\begin{align}\label{choicea}
\mathcal{E}_m:=\left\{\QQ(a)\in\mathcal{E}: |a|<\min\left\{\frac{\sqrt{2c_l}}{2},\sqrt{c_l-\rho^2}\right\}\right\}
\end{align}
is a set of equivalent martingale measures.
\end{enumerate}
\end{thm}
\begin{proof}
    See \cite[Theorem 3.1]{arxiv}, \cite[Observation 3.3]{arxiv} and \cite[Theorem 3.2]{arxiv}.
\end{proof}

\subsection{Moments of the variance and the Radon-Nikodym derivative}
The existence of the two following moments will be used to prove that the compensated Hawkes process is a martingale under a suitable family of risk neutral probability measures:
\begin{align}\label{expectations}
    \EE[(X_T^{(a)})^{2+\varepsilon_1}]<\infty \hspace{0.3cm}\text{and}\hspace{0.3cm}\EE\left[\left(\frac{1}{v_t}\right)^{1+\varepsilon_2}\right]<\infty,
\end{align}
where $\varepsilon_1,\varepsilon_2>0$ are arbitrarily small.

First, we prove the existence and integrability of all positive moments of the variance process. Essentially, this is a consequence of the fact that all positive moments of the standard Heston variance and the compound Hawkes process exist (see Lemma \ref{Anmomentcompound} in the Appendix).
   \begin{lem}\label{lemintegr}  Let $s\geq1$. Then, $\EE\left[v_t^s\right]<\infty$ for all $t\in[0,T]$ and $\int_0^T\EE\left[v_t^s\right]dt<\infty$.
   \end{lem}
   \begin{proof}
     See Lemma \ref{lemAintegrnpower} in the Appendix. 
   \end{proof}

Now, we check the existence and integrability of the $s$th negative moment of the variance process under the condition $2\kappa \Bar{v}>s\sigma^2$. Note that the larger the $s$, the larger the fraction $\frac{\kappa \Bar{v}}{\sigma^2}$ must be, that is, the product of the mean reversion speed and the long-term variance divided by the volatility of the variance. Informally, the larger the fraction $\frac{\kappa \Bar{v}}{\sigma^2}$ is, the less likely is that the variance will approach zero and that contributes to the existence of negative moments. 
 \begin{lem}\label{leminver}Let $s\geq1$. If $2\kappa \Bar{v}>s\sigma^2$, $\EE\left[\frac{1}{v_t^s}\right]<\infty$ for all $t\in[0,T]$ and $\int_0^T\EE\left[\frac{1}{v_t^s}\right]dt<\infty$. 
   \end{lem}
   \begin{proof}
       See Lemma \ref{lemAinver} in the Appendix. 
   \end{proof}

Finally, we study the existence of positive moments of the Radon-Nikodym derivative $\frac{d\QQ(a)}{d\PP}=X_T^{(a)}$, given in Theorem \ref{risk}. The proof boils down to check that expectations of the type
\begin{align*}
    \EE\left[\exp\left(A\int_0^T\frac{1}{v_u}du\right)\right] \hspace{0.3cm}\text{and}\hspace{0.3cm}  \EE\left[\exp\left(B(a)\int_0^Tv_udu\right)\right]
\end{align*}
are finite, where $A$ is a constant independent of $a$ and $B(a)$ is a constant depending on $a$. To ensure that the first expectation is finite we require the two following conditions on the model parameters $2\kappa\Bar{v}>\sigma^2$ and $\frac{1-\rho^2}{D(s^2-s)}\left(\frac{2\kappa\Bar{v}-\sigma^2}{2\sigma}\right)^2>1$ (see Lemma \ref{lemAinvexp} in the Appendix) where $s>1$ is the moment that we want to study. To check that the second expectation is finite we use Proposition \ref{novikov}.

\begin{lem}\label{nmoment} Let $\QQ(a)\in\mathcal{E}$, $s>1$, $D:=\sup_{t\in[0,T]}(\mu_t-r)^2<\infty$ and $X^{(a)}$ defined in Theorem \ref{risk}. Assume that $2\kappa\Bar{v}>\sigma^2$ and $\frac{1-\rho^2}{D(s^2-s)}\left(\frac{2\kappa\Bar{v}-\sigma^2}{2\sigma}\right)^2>1$. Consider $q_2$ such that $1<q_2<\frac{1-\rho^2}{D(s^2-s)}\left(\frac{2\kappa\Bar{v}-\sigma^2}{2\sigma}\right)^2$ and define $q_1:=\frac{q_2}{q_2-1}>1$. If
\begin{align}\label{acond}
    |a|<\min\Bigg\{\frac{1}{q_1s}\sqrt{\frac{c_l}{2}},\sqrt{\frac{(1-\rho^2)c_l}{q_1s\left[2q_1s(1-\rho^2)+\rho^2s-1\right]}}\Bigg\},
\end{align}
then
\begin{align*}
    \EE\left[\left(X_T^{(a)}\right)^s\right]<\infty \hspace{0.5cm}\text{and}\hspace{0.5cm} \EE\left[\left(X_t^{(a)}\right)^s\right]\leq \left(\frac{s}{s-1}\right)^s\EE\left[\left(X_T^{(a)}\right)^s\right]<\infty,
\end{align*}
for all $t\in[0,T]$.
\end{lem}
\begin{proof}
    See Lemma \ref{Anmoment} in the Appendix.
\end{proof}
\begin{obs}\label{obspositivity}
Using that $q_1,s>1$ and that $\rho^2<1$ one can check that the expression $q_1 s\left[2q_1 s(1-\rho^2)+\rho^2s-1\right]$ appearing in the second expression inside the minimum \eqref{acond} is strictly positive. Indeed,
\begin{align*}
    2q_1 s(1-\rho^2)+\rho^2s-1&>2(1-\rho^2)+\rho^2-1=1-\rho^2>0.
\end{align*}
\end{obs}

To ensure that the expectations in \eqref{expectations} are finite we assume some conditions on the model parameters and define a suitable family of risk neutral probability measures.
\begin{ass}\label{as3}
We fix $\varepsilon_1,\varepsilon_2>0$ and assume that
\begin{enumerate}
\item $\rho^2<c_l$. 
    \item $\frac{1-\rho^2}{D\left[(2+\varepsilon_1)^2-(2+\varepsilon_1)\right]}\left(\frac{2\kappa\Bar{v}-\sigma^2}{2\sigma}\right)^2>1$.
    \item $2\kappa\Bar{v}>(1+\varepsilon_2)\sigma^2$.
\end{enumerate}
\end{ass}
\begin{definition}
  Let $q,s>1$, we define a subset of $\mathcal{E}_m$ by
    \begin{align*}
        \mathcal{E}_{m}(q,s):=\left\{\QQ(a)\in\mathcal{E}_m:|a|<\min\Bigg\{\frac{1}{qs}\sqrt{\frac{c_l}{2}},\sqrt{\frac{(1-\rho^2)c_l}{qs\left[2qs(1-\rho^2)+\rho^2s-1\right]}}\Bigg\}\right\}.
    \end{align*}
\end{definition}
From now on, we fix $Q_2$ such that $1<Q_2<\frac{1-\rho^2}{D\left[(2+\varepsilon_1)^2-(2+\varepsilon_1)\right]}\left(\frac{2\kappa\Bar{v}-\sigma^2}{2\sigma}\right)^2$ and $Q_1:=\frac{Q_2}{Q_2-1}>1$. 
\begin{obs}\label{momentsobs}
Let $\QQ(a)\in\mathcal{E}_{m}(Q_1,2+\varepsilon_1)$. By Lemma \ref{leminver}, Lemma \ref{nmoment} and Assumption \ref{as3} the following holds
\begin{align*}
    \EE[(X_T^{(a)})^{2+\varepsilon_1}]<\infty,\hspace{0.3cm}\EE\left[\left(\frac{1}{v_t}\right)^{1+\varepsilon_2}\right]<\infty\hspace{0.3cm}\text{and}\hspace{0.3cm}\int_0^T\EE\left[\left(\frac{1}{v_t}\right)^{1+\varepsilon_2}\right]dt<\infty.
\end{align*}
\end{obs}

\subsection{Compensator of $N$ under the historic and the risk neutral measures}
First, we give the compensators of $N$ and $L$ under $\PP$. By definition, the compensator $\Lambda^N$ is a $\mathcal{F}^N$-predictable process such that $N-\Lambda^N$ is a $(\mathcal{F}^N,\PP)$-local martingale. Using the independence between $(B,W)$ and $N$ one can prove that martingale property of $N-\Lambda^N$ still holds under the joint filtration $\mathcal{F}=\{\mathcal{F}_t=\mathcal{F}_t^{(B,W)}\vee\mathcal{F}_t^L, t\in[0,T]\}$.
\begin{lem}\label{mainlemcomp}The following holds \begin{enumerate}
    \item Define $\Lambda^N_t:=\int_0^t\lambda_udu$, then $N-\Lambda^N$ is a square integrable $(\mathcal{F},\PP)$-martingale. 
\item Define $\Lambda^L_t:=\EE[J_1]\int_0^t\lambda_udu$, then $L-\Lambda^L$ is a square integrable $(\mathcal{F},\PP)$-martingale. 
\end{enumerate}
\end{lem}
\begin{proof}
(1) follows from the fact that $N-\Lambda^N$ is a $(\mathcal{F}^N,\PP)$-martingale, see \cite[Theorem 3]{Hawkesfinance}, the independence between $(B,W)$ and $N$, and $\EE[N_t^2]<\infty$, $\EE[\lambda_t^2]<\infty$, see \cite[Theorem 1]{momentshawkes} and \cite[Section 3.1]{momentshawkes}. An analogous proof can be employed to prove (2). 
\end{proof}
To conclude this section, we prove that the compensated (compound) Hawkes process is a martingale under  $\QQ(a)\in\mathcal{E}_{m}(Q_1,2+\varepsilon_1)$ using the existence of the moments in Observation \ref{momentsobs}.

\begin{prop}\label{qacomhawkesmain} Let $\QQ(a)\in\mathcal{E}_{m}(Q_1,2+\varepsilon_1)$, then
\begin{enumerate}
    \item $N-\Lambda^N$ is a $(\mathcal{F},\QQ(a))$-martingale.
    \item $L-\Lambda^L$ is a $(\mathcal{F},\QQ(a))$-martingale.
\end{enumerate}
\end{prop}
\begin{proof}
    See Proposition \ref{qacomhawkes} in the Appendix. 
\end{proof}
\begin{remark}
    Regarding Proposition \ref{qacomhawkesmain}, it is important to mention that the fact that $N-\Lambda^N$ is a $(\mathcal{F},\QQ(a))$-local martingale can be straightforwardly deduced by the Girsanov-Meyer Theorem \cite[Section III.8, Theorem 35]{Protter}. However, the martingale property of the process $N-\Lambda^N$ under $\QQ(a)\in\mathcal{E}_{m}(Q_1,2+\varepsilon_1)$ proven in Proposition \ref{qacomhawkesmain} is interesting on its own. This property is not obvious and the proof is rather technical, requiring the existence of moments of the variance $v$ and the Radon-Nikodym derivative $\frac{d\QQ(a)}{d\PP}=X_T^{(a)}$. Such existence of moments is an important issue in asset pricing theory, in relation to the existence of bubbles.  
\end{remark}

\section{Thiele's PIDE for unit-linked policies}\label{main}
The objective of this section is to derive Thiele's differential equation for unit-linked policies under the Heston-Hawkes stochastic volatility model. First, we find the drift of a process of the form $t\to Z(t,S_t,v_t,\lambda_t)$ under a risk neutral probability measure, where $Z$ is a regular enough function. The compensator of $N$ under $\QQ(a)\in\mathcal{E}_{m}(Q_1,2+\varepsilon_1)$ is needed to find such drift. In order to lighten the notation, we first define some space of functions. From now on, $\RR_+:=(0,\infty)$.
\begin{definition}
    We define $\mathcal{D}:=\RR_+^2\times[\lambda_0,\infty)$ and $\mathcal{C}^{1,2}:=\mathcal{C}^{1,2}\left([0,T]\times\mathcal{D}\right)$ the space of functions $Y\colon[0,T]\times\mathcal{D}\to\RR_+$ that are jointly continuous, continuously differentiable on the first variable, twice continuously differentiable on the last three variables and all derivatives are jointly continuous. 
\newline
We define $\mathcal{C}^{0,1,2}:=\mathcal{C}^{0,1,2}\left([0,T]^2\times\mathcal{D}\right)$ the space of functions $Z\colon[0,T]^2\times\mathcal{D}\to\RR_+$ that are jointly continuous, continuous on the first variable, continuously differentiable on the second variable, twice continuously differentiable on the last three variables and all derivatives are jointly continuous. 
\end{definition}

Let $\QQ(a)\in\mathcal{E}_m$ and $\varphi\colon[0,T]\times\RR_+\to\RR_+$ a payoff such that $\EE^{\QQ(a)}[|\varphi(s,S_s)|]<\infty$ for all $s\in[0,T]$. Recall that the price at time $t\in[0,T]$ of the payoff function $\varphi(s,S_s)$ with maturity $s\in[t,T]$ is given by $e^{-r(s-t)}\EE^{\QQ(a)}[\varphi(s,S_s)| \mathcal{F}_t]$. One example of such payoff is $\varphi(s,S_s)=\max\{G,S_s\}$, where $G$ is called the guaranteed amount. This means that at time $s$ the insured is paid the maximum between the guaranteed amount and the stock.

As a consequence of the Markov property of the process $(N,\lambda)$, see \cite[Remark 1.22]{ETH}, we prove in the next lemma that $\EE^{\QQ(a)}[\varphi(s,S_s)| \mathcal{F}_t]$ is a deterministic function of the joint process $(t,S_t,v_t,\lambda_t)$. Due to the presence of jumps in $v$ and $\lambda$, such function is the solution of a partial integro-differential equation, PIDE from now on. By applying Itô's formula and using the compensators of $N$ and $L$ under $\QQ(a)\in\mathcal{E}_{m}(Q_1,2+\varepsilon_1)$ we obtain that PIDE.

\begin{lem}\label{lema1} Let $\QQ(a)\in\mathcal{E}_{m}(Q_1,2+\varepsilon_1)$, $\varphi\colon[0,T]\times\RR_+\to\RR_+$  such that $\EE^{\QQ(a)}[|\varphi(s,S_s)|]<\infty$ for all $s\in[0,T]$. Then, there exists a function $Z^{\varphi,a}\colon[0,T]^2\times\mathcal{D}\to\RR_+$ such that 
\begin{align*}
    \EE^{\QQ(a)}[\varphi(s,S_s) | \mathcal{F}_t ]=Z_{s}^{\varphi,a}(t,S_t,v_t,\lambda_t),
\end{align*}
where $s,t\in[0,T]$. Note that $Z^{\varphi,a}_s(t,x,y,z)=\varphi(s,x)$ for $t\in[s,T]$.

Furthermore, fix $s\in[0,T]$, if $Z_s^{\varphi,a}\in\mathcal{C}^{1,2}$, it satisfies the following PIDE
\begin{gather}
    \partial_t Z^{\varphi,a}_s(t,x,y,z)+rx\partial_xZ^{\varphi,a}_s(t,x,y,z)-\kappa^{(a)}(y-\Bar{v}^{(a)})\partial_yZ^{\varphi,a}_s(t,x,y,z)\nonumber\\
    -\beta(z-\lambda_0)\partial_zZ^{\varphi,a}_s(t,x,y,z)+\frac{1}{2}x^2y\partial_{xx}^2Z^{\varphi,a}_s(t,x,y,z)+\frac{1}{2}\sigma^2y\partial_{yy}^2Z^{\varphi,a}_s(t,x,y,z)\nonumber\\
    +\sigma\rho xy\partial_{xy}^2Z^{\varphi,a}_s(t,x,y,z)+z\int_{(0,\infty)}\left[Z^{\varphi,a}_s(t,x,y+\eta u, z+\alpha)-Z^{\varphi,a}_s(t,x,y,z)\right]P_{J_1}(du)=0, \label{eq:pid}
\end{gather}
for $(t,x,y,z)\in[0,s]\times\mathcal{D}$, final condition $Z^{\varphi,a}_s(s,x,y,z)=\varphi(s,x)$ and $P_{J_1}$ is the law of $J_1$. 
\end{lem}
\begin{proof}
Since $(N,\lambda)$ is a Markov process, see \cite[Remark 1.22]{ETH}, $(S,v,\lambda)$ is also a Markov process and we conclude that there exists a function $Z^{\varphi,a}\colon[0,T]^2\times\mathcal{D}\to\RR_+$ such that
\begin{align*}
    \EE^{\QQ(a)}[\varphi(s,S_s) | \mathcal{F}_t ]=Z^{\varphi,a}_s(t,S_t,v_t,\lambda_t),
\end{align*}
where $s,t\in[0,T]$. 

Moreover, fix $s\in[0,T]$, if $Z_s^{\varphi,a}\in\mathcal{C}^{1,2}$ we can apply Itô formula to the process $t\to Z^{\varphi,a}_s(t,S_t,v_t,\lambda_t)$ for $t\in[0,s]$. For convenience, we define $Y_t:=(t,S_t,v_t,\lambda_t)$. Applying Itô formula to $Z^{\varphi,a}_s$ we get 
\begin{align}\label{ito2}
    Z^{\varphi,a}_s(Y_t)= \ & Z^{\varphi,a}_s(Y_0)+\int_0^t\Big[\partial_tZ^{\varphi,a}_s(Y_{u})+rS_u\partial_xZ^{\varphi,a}_s(Y_{u})-\kappa^{(a)}(v_u-\Bar{v}^{(a)})\partial_yZ^{\varphi,a}_s(Y_{u})\notag\\
    &-\beta(\lambda_u-\lambda_0)\partial_zZ^{\varphi,a}_s(Y_{u}) +\frac{1}{2}S_u^2v_u\partial_{xx}^2Z^{\varphi,a}_s(Y_u)+\frac{1}{2}\sigma^2v_u\partial_{yy}^2Z^{\varphi,a}_s(Y_u) \notag\\
    &+\sigma\rho S_uv_u\partial_{xy}^2Z_{s}^{\varphi,a}(Y_u)\Big]du+\sqrt{1-\rho^2}\int_0^tS_{u}\sqrt{v_{u}}\partial_xZ^{\varphi,a}_s(Y_{u-})dB_u^{\QQ(a)}\notag\\
    &+\int_0^t\left[\rho S_{u}\sqrt{v_{u}}\partial_xZ^{\varphi,a}_s(Y_{u-})+\sigma\sqrt{v_{u}}\partial_yZ^{\varphi,a}_s(Y_{u-})\right]dW_u^{\QQ(a)} \notag\\
    &+\sum_{0<u\leq t}\left[Z^{\varphi,a}_s(Y_u)-Z^{\varphi,a}_s(Y_{u-})\right].
\end{align}
Next, we can write 
\begin{align*}
    \sum_{0<u\leq t}\left[Z^{\varphi,a}_s(Y_u)-Z^{\varphi,a}_s(Y_{u-})\right] & = \sum_{0<u\leq t} g^{\varphi,a}_s(u,\Delta L_u,\Delta N_u),
\end{align*}
where 
\begin{align*}
    g^{\varphi,a}_s(u,b_1,b_2):=Z^{\varphi,a}_s(u,S_u,v_{u-}+\eta b_1, \lambda_{u-}+\alpha b_2)-Z^{\varphi,a}_s(u,S_u,v_{u-},\lambda_{u-}).
\end{align*} 
We now define $M_u=(L_u,N_u)$ and for $t\in[0,T]$, $B\in\mathcal{B} (\R^2\setminus\{0,0\})$ 
\begin{align*}
    N^M(t,A)=\#\{0<s\leq t, \Delta M_s\in A\}.
\end{align*}
We add and subtract the compensator of the counting measure $N^M$ to split the expression into a $(\mathcal{F},\QQ(a))$-local martingale plus a predictable process of finite variation. 
By Proposition \ref{qacomhawkesmain}, the compensators of $N$ and $L$ under $\QQ(a)\in\mathcal{E}_{m}(Q_1,2+\varepsilon_1)$ are $\Lambda_t^N=\int_0^t\lambda_udu$ and $\Lambda_t^L=\EE[J_1]\int_0^t\lambda_udu$ respectively. Thus,
\begin{align*}
    \sum_{0<u\leq t}\left[Z^{\varphi,a}_s(Y_u)-Z^{\varphi,a}_s(Y_{u-})\right]  = \ & \int_0^t\int_{(0,\infty)^2}g^{\varphi,a}_s(u,b_1,b_2)N^M(du,db) \\
     = \ & \int_0^t\int_{(0,\infty)^2} g^{\varphi,a}_s(u,b_1,b_2)\left(N^M(du,db)-\lambda_uP_{J_1}(db_1)\delta_1(db_2)du\right) \\ &+\int_0^t\lambda_u\int_{(0,\infty)} g^{\varphi,a}_s(u,b_1,1)P_{J_1}(db_1)du.
\end{align*}
Replacing everything in \eqref{ito2} we finally get 
\begin{align}\label{martingala}
    Z^{\varphi,a}_s(Y_t)= \ & Z^{\varphi,a}_s(Y_0)+\int_0^t\Big[\partial_tZ^{\varphi,a}_s(Y_{u})+rS_u\partial_xZ^{\varphi,a}_s(Y_{u})-\kappa^{(a)}(v_u-\Bar{v}^{(a)})\partial_yZ^{\varphi,a}_s(Y_{u}) \notag\\
    &-\beta(\lambda_u-\lambda_0)\partial_zZ^{\varphi,a}_s(Y_{u})+\frac{1}{2}S_u^2v_u\partial_{xx}^2Z^{\varphi,a}_s(Y_u)+\frac{1}{2}\sigma^2v_u\partial_{yy}^2Z^{\varphi,a}_s(Y_u)\notag \\
    &+\sigma\rho S_uv_u\partial_{xy}^2Z_{s}^{\varphi,a}(Y_u)+\lambda_u\int_{(0,\infty)} g_s^{\varphi,a}(u,b_1,1)P_{J_1}(db_1)\Big]du \notag\\
    &+\sqrt{1-\rho^2}\int_0^tS_{u}\sqrt{v_{u}}\partial_xZ^{\varphi,a}_s(Y_{u-})dB_u^{\QQ(a)}\notag\\
    &+\int_0^t\left[\rho S_{u}\sqrt{v_{u}}\partial_xZ^{\varphi,a}_s(Y_{u-})+\sigma\sqrt{v_{u}}\partial_yZ^{\varphi,a}_s(Y_{u-})\right]dW_u^{\QQ(a)} \notag\\
    &+\int_0^t\int_{(0,\infty)^2} g^{\varphi,a}_s(u,b_1,b_2)\left(N^M(du,db)-\lambda_uP_{J_1}(db_1)\delta_1(db_2)du\right).
\end{align}
Recall that $t\to Z^{\varphi,a}_s(Y_t)$ is a $(\mathcal{F},\QQ(a))$-martingale and note that the last three terms in \eqref{martingala} are $(\mathcal{F},\QQ(a))$-local martingales. Moving this three terms to the left hand side we see that a local martingale is equal to a continuous process of finite variation. This implies that the integral of the drift is 0 on every interval $[0,t]\subset[0,T]$ and that the sum of the last three terms in \eqref{martingala} is a $(\mathcal{F},\QQ(a))$-martingale. As a consequence, the drift is constant equal to $0$ and $Z^{\varphi,a}_s$ satisfies the following PIDE
\begin{gather*}
    \partial_t Z^{\varphi,a}_s(t,x,y,z)+rx\partial_xZ^{\varphi,a}_s(t,x,y,z)-\kappa^{(a)}(y-\Bar{v}^{(a)})\partial_yZ^{\varphi,a}_s(t,x,y,z)\notag\\
    -\beta(z-\lambda_0)\partial_zZ^{\varphi,a}_s(t,x,y,z)+\frac{1}{2}x^2y\partial_{xx}^2Z^{\varphi,a}_s(t,x,y,z)+\frac{1}{2}\sigma^2y\partial_{yy}^2Z^{\varphi,a}_s(t,x,y,z)\notag\\
    +\sigma\rho xy\partial_{xy}^2Z^{\varphi,a}_s(t,x,y,z)+z\int_{(0,\infty)}\left[Z^{\varphi,a}_s(t,x,y+\eta u, z+\alpha)-Z^{\varphi,a}_s(t,x,y,z)\right]P_{J_1}(du)=0,
\end{gather*}
for $(t,x,y,z)\in[0,s]\times\mathcal{D}$ and final condition $Z^{\varphi,a}_s(s,x,y,z)=\varphi(s,x)$. Note that $\mathcal{D}$ is the support of the process $(S_t, v_t, \lambda_t)$. 
\end{proof}

\begin{definition}\label{loperator} Let $\QQ(a)\in\mathcal{E}_{m}(Q_1,2+\varepsilon_1)$, $f\in\mathcal{C}^{1,2}$ satisfying
\begin{align}\label{intJ1}
    \int_{(0,\infty)}|f(t,x,y+\eta u, z+\alpha)|P_{J_1}(du)<\infty.
\end{align}
We define the following partial integro-differential operator $\mathcal{L}^a$ by
\begin{align*}
    \mathcal{L}^af(t,x,y,z):= \ & rx\partial_xf(t,x,y,z)-\kappa^{(a)}(y-\Bar{v}^{(a)})\partial_yf(t,x,y,z)-\beta(z-\lambda_0)\partial_zf(t,x,y,z) \\
    &+\frac{1}{2}x^2y\partial_{xx}^2f(t,x,y,z)+\frac{1}{2}\sigma^2y\partial_{yy}^2f(t,x,y,z)+\sigma\rho xy\partial_{xy}^2f(t,x,y,z) \\
    &+z\int_{(0,\infty)}\left[f(t,x,y+\eta u, z+\alpha)-f(t,x,y,z)\right]P_{J_1}(du),
\end{align*}
where $P_{J_1}$ is the law of $J_1$. 
\end{definition}
\begin{obs}
Since all positive moments of $J_1$ exists, condition \eqref{intJ1} automatically holds if $f$ has polynomial growth in the third variable, that is, if $f(t,x,y,z)\leq \sum_{i=0}^nC_i(t,x,z)y^i$ where the constants $C_i(t,x,z)$ do not depend on $y$. Additionally, if $|f(t,x,y+\eta u,z+\alpha)|\leq C(t,x,y,z+\alpha)e^{cu}$ with $c<\epsilon_J$, then condition \eqref{intJ1} holds. 
\end{obs}
\begin{obs}\label{obs} Let $\QQ(a)\in\mathcal{E}_{m}(Q_1,2+\varepsilon_1)$ and $f\in\mathcal{C}^{1,2}$. Note that we have proved in Lemma \ref{lema1} that the drift of the Itô differential of $f(t,S_t,v_t,\lambda_t)$ is $\partial_tf(t,S_t,v_t,\lambda_t)+\mathcal{L}^af(t,S_t,v_t,\lambda_t)$. Furthermore, the PIDE in \eqref{eq:pid} is just 
\begin{align}\label{pidcompact}
    \partial_tZ^{\varphi,a}_{s}(t,x,y,z)+\mathcal{L}^aZ^{\varphi,a}_s(t,x,y,z)=0,
\end{align}
for $(t,x,y,z)\in[0,s]\times\mathcal{D}$ and final condition $Z^{\varphi,a}_s(s,x,y,z)=\varphi(s,x)$
\end{obs}

From the PIDE obtained in Lemma \ref{lema1} we derive the PIDE that satisfies $e^{-r(s-t)}\EE^{\QQ(a)}[\varphi(s,S_s)| \mathcal{F}_t]$, which is the price at time $t$ of the payoff  $\varphi(s,S_s)$ where $0\leq t\leq s \leq T$.

\begin{lem}\label{lem2} Let $\QQ(a)\in\mathcal{E}_{m}(Q_1,2+\varepsilon_1)$, $\varphi\colon[0,T]\times\RR_+\to\RR_+$ such that $\EE^{\QQ(a)}[|\varphi(s,S_s)|]<\infty$ for all $s\in[0,T]$. Then, there exists a function $U^{\varphi,a}\colon[0,T]^2\times\mathcal{D}\to\RR_+$ such that
\begin{align}\label{relation}
    e^{-r(s-t)}\EE^{\QQ(a)}[\varphi(s,S_s) | \mathcal{F}_t ]=U_s^{\varphi,a}(t,S_t,v_t,\lambda_t)
\end{align}
where $s,t\in[0,T]$. Note that $U_s^{\varphi,a}(t,x,y,z)=e^{-r(s-t)}\varphi(s,x)$ for $t\in[s,T]$. 
\newline
Furthermore, fix  $s\in[0,T]$, if $U_s^{\varphi,a}\in\mathcal{C}^{1,2}$, it satisfies the following PIDE
\begin{align}\label{pid2}
    \partial_t U_s^{\varphi,a}(t,x,y,z)+\mathcal{L}^aU_s^{\varphi,a}(t,x,y,z)=rU_s^{\varphi,a}(t,x,y,z),
\end{align}
where $\mathcal{L}^a$ is defined in Definition \ref{loperator}, $(t,x,y,z)\in[0,s]\times\mathcal{D}$ and final condition $U_s^{\varphi,a}(s,x,y,z)=\varphi(s,x)$.
\end{lem}
\begin{proof}
    See Lemma \ref{lem2A} in the Appendix. 
\end{proof}

We now have all the preliminary results needed to obtain Thiele's PIDE for unit-linked policies under the Heston-Hawkes stochastic volatility model. First, we introduce the insurance model. Let $\mathcal{X}=\{\mathcal{X}_t,t\in[0,T]\}$ be a regular Markov chain with finite state space $\mathcal{J}$ that describes the insured's state. We write $p_{ij}(s,t)$, $i,j\in\mathcal{J}$, $0\leq s\leq t\leq T$ for the transition probabilities and $\mu_{ij}(t),\mu_i(t)$ for the transition rates, $i,j\in\mathcal{J}$, $t\in[0,T]$. Let $\QQ(a)\in\mathcal{E}_{m}(Q_1,2+\varepsilon_1)$ and $f_j,g_j,h_{jk}\colon[0,T]\times\RR_+\to\RR_+$ be policy functions where $j,k\in\mathcal{J}$, $j\neq k$,  satisfying $\EE^{\QQ(a)}[f_j(t,S_t)]<\infty$, $\EE^{\QQ(a)}[g_j(t,S_t)]<\infty$, $\EE^{\QQ(a)}[h_{jk}(t,S_t)]<\infty$ for all $t\in[0,T]$, 
\begin{align*}
    \int_0^Te^{-rs}p_{ij}(t,s)\EE^{\QQ(a)}[g_j(s,S_s) | \mathcal{F}_t ]ds <\infty,
\end{align*}
and
\begin{align*}
    \int_0^Te^{-rs}p_{ij}(t,s)\mu_{jk}(s)\EE^{\QQ(a)}[h_{jk}(s,S_s) | \mathcal{F}_t ]ds<\infty,
\end{align*}
a.s. for all $j,k\in\mathcal{J}$, $j\neq k$. The mathematical reserve $V_{i,\mathcal{F}}^{+,a}(t)$ of a contract with policy functions $f_j$, $g_j$ and $h_{jk}$ given that the insured is in state $i\in\mathcal{J}$ at time $t$ and the information $\mathcal{F}_t$, is given by
\begin{align}\label{reserve}
    V_{i,\mathcal{F}}^{+,a}(t)= \ & e^{rt}\Bigg[\sum_{j\in \mathcal{J}}e^{-rT}p_{ij}(t,T)\EE^{\QQ(a)}[f_j(T,S_T) | \mathcal{F}_t ] \notag\\ 
    &+\sum_{j\in\mathcal{J}}\int_t^Te^{-rs}p_{ij}(t,s)\EE^{\QQ(a)}[g_j(s,S_s) | \mathcal{F}_t ]ds \notag\\
    &+\sum_{\substack{j,k\in \mathcal{J}\\ j\neq k}}\int_t^Te^{-rs}p_{ij}(t,s)\mu_{jk}(s)\EE^{\QQ(a)}[h_{jk}(s,S_s) | \mathcal{F}_t ] ds\Bigg].
\end{align}
See \cite{Ragnar} for a reference on mathematical reserves in life insurance. In the following result we derive Thiele's PIDE. 
\begin{prop}\label{thiele}(Thiele's PIDE) Let $\QQ(a)\in\mathcal{E}_{m}(Q_1,2+\varepsilon_1)$ and $i\in\mathcal{J}$, then, there exists a function $V_i^a\colon[0,T]\times\mathcal{D}\to\RR_+$ such that
    \begin{align*}
        V_{i,\mathcal{F}}^{+,a}(t)=V_i^a(t,S_t,v_t,\lambda_t).
    \end{align*}
    Furthermore, assume that the functions $U_T^{f_j,a},U^{g_j,a},U^{h_{jk},a}$ given by Lemma \ref{lem2} satisfy the following $U_T^{f_j,a}\in\mathcal{C}^{1,2}$ and  $U^{g_j,a},U^{h_{jk},a}\in\mathcal{C}^{0,1,2}$ for all $j,k\in\mathcal{J}$, $j\neq k$. Then, $V_i^a$ satisfies the following PIDE
    \begin{align*}
        \partial_tV_i^a(t,x,y,z) = \ &  rV_i^a(t,x,y,z)-g_i(t,x)-\sum_{\substack{k\in \mathcal{J}\\ k\neq i}}\mu_{ik}(t)\left(h_{ik}(t,x)+V_k^a(t,x,y,z)-V_i^a(t,x,y,z)\right) \\
        & -\mathcal{L}^aV_i^a(t,x,y,z),
    \end{align*}
    where $\mathcal{L}^a$ is the operator defined in Definition \ref{loperator}, $(t,x,y,z)\in[0,T]\times\mathcal{D}$ and final condition $V_i^a(T,x,y,z)=f_i(T,x)$.
    \end{prop}
\begin{proof}
Applying Lemma \ref{lem2} there exist functions $U^{f_j,a},U^{g_j,a},U^{h_{jk},a}\colon[0,T]^2\times\mathcal{D}\to\RR_+$ for all $j,k\in\mathcal{J}$, $j\neq k$ such that
the mathematical reserve $V_{i,\mathcal{F}}^{+,a}(t)$ in \eqref{reserve} can be rewritten as
\begin{align}\label{res2}
     V_{i,\mathcal{F}}^{+,a}(t)= \ & \sum_{j\in \mathcal{J}}p_{ij}(t,T)U_T^{f_j,a}(t,S_t,v_t,\lambda_t) +\sum_{j\in\mathcal{J}}\int_t^Tp_{ij}(t,s)U_s^{g_j,a}(t,S_t,v_t,\lambda_t)ds \notag\\
    &+\sum_{\substack{j,k\in \mathcal{J}\\ j\neq k}}\int_t^Tp_{ij}(t,s)\mu_{jk}(s)U_s^{h_{jk},a}(t,S_t,v_t,\lambda_t) ds.
\end{align}
Defining $V_i^{(a)}\colon[0,T]\times\mathcal{D}\to\RR_+$ by
\begin{align*}
     V_i^{(a)}(t,x,y,z):= \ & \sum_{j\in \mathcal{J}}p_{ij}(t,T)U_T^{f_j,a}(t,x,y,z) +\sum_{j\in\mathcal{J}}\int_t^Tp_{ij}(t,s)U_s^{g_j,a}(t,x,y,z)ds \notag\\
    &+\sum_{\substack{j,k\in \mathcal{J}\\ j\neq k}}\int_t^Tp_{ij}(t,s)\mu_{jk}(s)U_s^{h_{jk},a}(t,x,y,z) ds.
\end{align*}
we see that $V_{i,\mathcal{F}}^{+,a}(t)=V_i^a(t,S_t,v_t,\lambda_t)$ and the first part is proved. 

Assume now that $U_T^{f_j,a}\in\mathcal{C}^{1,2}$ and  $U^{g_j,a},U^{h_{jk},a}\in\mathcal{C}^{0,1,2}$ for all $j,k\in\mathcal{J}$, $j\neq k$. Since $\mathcal{X}$ is regular, we see that $\partial_tV_i^a$ and $\mathcal{L}^aV_i^a$ are well defined by applying several times differentiation under the integral sign. For the sake of clarity, we define
    \begin{align}\label{defV}
        V_{i,\mathcal{F}}^{+,a}(t)=V_i^a(t,S_t,v_t,\lambda_t)=G_{i,T}^{a}(t,S_t,v_t,\lambda_t)+\int_t^TF_{i,s}^{a}(t,S_t,v_t,\lambda_t)ds,
    \end{align}
    where
    \begin{align}
        G_{i,T}^{a}(t,x,y,z):&=\sum_{j\in \mathcal{J}}p_{ij}(t,T)U_T^{f_j,a}(t,x,y,z), \notag\\
        F_{i,s}^{a}(t,x,y,z):&=\sum_{j\in \mathcal{J}}p_{ij}(t,s)U_s^{\theta_j,a}(t,x,y,z),\label{defF} \\
        \theta_j(s,x):&=g_j(s,x)+\sum_{\substack{k\in \mathcal{J}\\ k\neq j}}\mu_{jk}(s)h_{jk}(s,x). \notag
    \end{align}
Since $\mathcal{X}$ is regular, note that $U^{\theta_j,a},F_i^{a}\in\mathcal{C}^{0,1,2}$. Fix $s\in[0,T]$, we can apply Itô's formula to the processes $t\to U_s^{\theta_j,a}(t,S_t,v_t,\lambda_t)$ and $t\to F_{i,s}^{a}(t,S_t,v_t,\lambda_t)$. Now, we compute the drift of the Itô differential of $t\to F_{i,s}^{a}(t,S_t,v_t,\lambda_t)$ in two ways. First, by direct definition. By Observation \ref{obs}, the drift of the Itô differential of $t\to F_{i,s}^{a}(t,S_t,v_t,\lambda_t)$ is
\begin{align}\label{drift1}
    \partial_tF_{i,s}^{a}(t,S_t,v_t,\lambda_t)+\mathcal{L}^aF_{i,s}^{a}(t,S_t,v_t,\lambda_t).
\end{align}
On the other hand
\begin{align*}
    dF_{i,s}^{a}(t,S_t,v_t,\lambda_t)=\sum_{j\in\mathcal{J}}\partial_tp_{ij}(t,s)U_s^{\theta_j,a}(t,S_t,v_t,\lambda_t)dt+\sum_{j\in\mathcal{J}}p_{ij}(t,s)dU_s^{\theta_j,a}(t,S_t,v_t,\lambda_t).
\end{align*}
Using Kolmogorov's backward equation in the first term and then the definition of $F_{i,s}^{a}(t,x,y,z)$ given in \eqref{defF} we get
\begin{align*}
     dF_{i,s}^{a}(t,S_t,v_t,\lambda_t) = \ & \sum_{j\in\mathcal{J}}\sum_{\substack{k\in \mathcal{J}\\ k\neq i}}\mu_{ik}(t)\left(p_{ij}(t,s)-p_{kj}(t,s)\right)U_s^{\theta_j,a}(t,S_t,v_t,\lambda_t)dt \\ 
     & +\sum_{j\in\mathcal{J}}p_{ij}(t,s)dU_s^{\theta_j,a}(t,S_t,v_t,\lambda_t) \\
     = \ & \sum_{\substack{k\in \mathcal{J}\\ k\neq i}}\mu_{ik}(t)\left(F_{i,s}^{a}(t,S_t,v_t,\lambda_t)-F_{k,s}^{a}(t,S_t,v_t,\lambda_t)\right)dt \\ 
     & +\sum_{j\in\mathcal{J}}p_{ij}(t,s)dU_s^{\theta_j,a}(t,S_t,v_t,\lambda_t).
\end{align*}
We know by Observation \ref{obs} that the drift part of the Itô differential of $t\to U_s^{\theta_j,a}(t,S_t,v_t,\lambda_t)$ is $\partial_tU_s^{\theta_j,a}(t,S_t,v_t,\lambda_t)+\mathcal{L}^aU_s^{\theta_j,a}(t,S_t,v_t,\lambda_t)$. Moreover, $U_s^{\theta_j,a}(t,S_t,v_t,\lambda_t)$ satisfies the PIDE in \eqref{pid2}. Thus, the drift part of the Itô differential of $t\to F_{i,s}^{a}(t,S_t,v_t,\lambda_t)$ can also be written as 
\begin{align}\label{drift2}
     \ \ & \sum_{\substack{k\in \mathcal{J}\\ k\neq i}}\mu_{ik}(t)\left(F_{i,s}^{a}(t,S_t,v_t,\lambda_t)-F_{k,s}^{a}(t,S_t,v_t,\lambda_t)\right) \notag\\ 
     &+\sum_{j\in\mathcal{J}}p_{ij}(t,s)\left(\partial_tU_s^{\theta_j,a}(t,S_t,v_t,\lambda_t)+\mathcal{L}^aU_s^{\theta_j,a}(t,S_t,v_t,\lambda_t)\right) \notag\\ 
     = \ & \sum_{\substack{k\in \mathcal{J}\\ k\neq i}}\mu_{ik}(t)\left(F_{i,s}^{a}(t,S_t,v_t,\lambda_t)-F_{k,s}^{a}(t,S_t,v_t,\lambda_t)\right) +\sum_{j\in\mathcal{J}}p_{ij}(t,s)rU_s^{\theta_j,a}(t,S_t,v_t,\lambda_t) \notag\\
     = \ & \sum_{\substack{k\in \mathcal{J}\\ k\neq i}}\mu_{ik}(t)\left(F_{i,s}^{a}(t,S_t,v_t,\lambda_t)-F_{k,s}^{a}(t,S_t,v_t,\lambda_t)\right) +rF_{i,s}^{a}(t,S_t,v_t,\lambda_t).
\end{align}
In the last step we have used again the definition of $F_{i,s}^{a}$ in \eqref{defF}. Equating the two equivalent expressions of the drift of the Itô differential in \eqref{drift1} and \eqref{drift2} we get
\begin{gather*}
    \partial_tF_{i,s}^{a}(t,S_t,v_t,\lambda_t)+\mathcal{L}^aF_{i,s}^{a}(t,S_t,v_t,\lambda_t) \\
    = \sum_{\substack{k\in \mathcal{J}\\ k\neq i}}\mu_{ik}(t)\left(F_{i,s}^{a}(t,S_t,v_t,\lambda_t)-F_{k,s}^{a}(t,S_t,v_t,\lambda_t)\right) +rF_{i,s}^{a}(t,S_t,v_t,\lambda_t).
\end{gather*}
We deduce that 
\begin{gather}
    \partial_tF_{i,s}^{a}(t,x,y,z)+\mathcal{L}^aF_{i,s}^{a}(t,x,y,z) \nonumber\\
    = \sum_{\substack{k\in \mathcal{J}\\ k\neq i}}\mu_{ik}(t)\left(F_{i,s}^{a}(t,x,y,z)-F_{k,s}^{a}(t,x,y,z)\right) +rF_{i,s}^{a}(t,x,y,z), \label{eq1}
\end{gather}
for $(t,x,y,z)\in[0,s]\times\mathcal{D}$. Since $F_{i}^{a}\in\mathcal{C}^{0,1,2}$ we can apply differentiation under the integral sign to get 
\begin{align}\label{derV}
    \partial_t\left(\int_t^TF_{i,s}^{a}(t,x,y,z)ds\right)=\int_t^T\partial_tF_{i,s}^{a}(t,x,y,z)ds-F_{i,t}^{a}(t,x,y,z).
\end{align}
Taking the derivative with respect to $t$ in \eqref{defV} and using \eqref{derV} we get
\begin{align*}
    \partial_tV_i^a(t,x,y,z)=\partial_tG_{i,T}^{a}(t,x,y,z) + \int_t^T\partial_tF_{i,s}^{a}(t,x,y,z)ds-F_{i,t}^{a}(t,x,y,z).
\end{align*}
Writing the explicit expression of $F_{i,t}^{a}(t,x,y,z)$ we obtain
\begin{align}\label{expr1}
    \int_t^T\partial_tF_{i,s}^{a}(t,x,y,z)ds=\partial_tV_i^a(t,x,y,z)-\partial_tG_{i,T}^{a}(t,x,y,z)+g_i(t,x)+\sum_{\substack{k\in \mathcal{J}\\ k\neq i}} \mu_{ik}(t)h_{ik}(t,x).
\end{align}
We now integrate \eqref{eq1} with respect to $s$ on the region $[t,T]$ to obtain
\begin{gather*}
    \int_t^T\partial_t^TF_{i,s}^{a}(t,x,y,z)ds+\int_t^T\mathcal{L}^aF_{i,s}^{a}(t,x,y,z)ds  \\
    = \sum_{\substack{k\in \mathcal{J}\\ k\neq i}}\mu_{ik}(t)\left(\int_t^TF_{i,s}^{a}(t,x,y,z)ds-\int_t^TF_{k,s}^{a}(t,x,y,z)\right)+r\int_t^TF_{i,s}^{a}(t,x,y,z)ds.
\end{gather*}
Since $F_{i,s}^{a}\in\mathcal{C}^{1,2}$ and it is positive we can apply differentiation under the integral sign several times and Tonelli's theorem to conclude that $\int_t^T\mathcal{L}^aF_{i,s}^{a}(t,x,y,z)ds=\mathcal{L}^a\int_t^TF_{i,s}^{a}(t,x,y,z)ds$. Now, we write the expression of $\int_t^T\partial_t^TF_{i,s}^{a}(t,x,y,z)$ in \eqref{expr1}, we use that $\int_t^T\mathcal{L}^aF_{i,s}^{a}(t,x,y,z)ds=\mathcal{L}^a\int_t^TF_{i,s}^{a}(t,x,y,z)ds$ and that $\int_t^TF_{i,s}^{a}(t,x,y,z)ds=V_i^a(t,x,y,z)-G_{i,T}^{a}(t,x,y,z)$ to get
\begin{gather}
 \partial_t\left(V_i^a(t,x,y,z)-G_{i,T}^{a}(t,x,y,z)\right)+g_i(t,x)+\sum_{\substack{k\in \mathcal{J}\\ k\neq i}} \mu_{ik}(t)h_{ik}(t,x)\nonumber \\ 
 +\mathcal{L}^a\left(V_i^a(t,x,y,z)-G_{i,T}^{a}(t,x,y,z)\right)  \nonumber \\
    = \ \sum_{\substack{k\in \mathcal{J}\\ k\neq i}}\mu_{ik}(t)\left(V_i^a(t,x,y,z)- V_k^a(t,x,y,z)+G_{k,T}^a(t,x,y,z)-G_{i,T}^{a}(t,x,y,z)\right) \nonumber \\
    +r\left(V_i^a(t,x,y,z)-G_{i,T}^{a}(t,x,y,z)\right). \label{eq:prethiele}
\end{gather}
Now we prove that the terms involving $G_T^{a}$ will cancel each other. Indeed, observe that using Kolmogorov's backward equation we have
\begin{align*}
    \partial_tG_{i,T}^{a}(t,x,y,z)  = \ & \sum_{j\in\mathcal{J}}\partial_tp_{ij}(t,T)U_T^{f_j,a}(t,x,y,z)+\sum_{j\in\mathcal{J}}p_{ij}(t,T)\partial_tU_T^{f_j,a}(t,x,y,z) \\
     = \ & \sum_{j\in\mathcal{J}} \sum_{\substack{k\in \mathcal{J}\\ k\neq i}}\mu_{ik}(t)\left(p_{ij}(t,s)-p_{kj}(t,s)\right)U_T^{f_j,a}(t,x,y,z) \\
    &+\sum_{j\in\mathcal{J}}p_{ij}(t,T)\partial_tU_T^{f_j,a}(t,x,y,z) \\
     = \ &\sum_{\substack{k\in \mathcal{J}\\ k\neq i}}\mu_{ik}(t)\left(G_{i,T}^{a}(t,x,y,z)-G_{k,T}^a(t,x,y,z)\right)+\sum_{j\in\mathcal{J}}p_{ij}(t,T)\partial_tU_T^{f_j,a}(t,x,y,z).
\end{align*}
Moreover, using that $\mathcal{L}^aG_{i,T}^{a}(t,x,y,z)=\sum_{j\in\mathcal{J}}p_{ij}(t,T)\mathcal{L}^aU_T^{f_j,a}(t,x,y,z)$ and that $U_T^{f_j,a}$ satisfies the PIDE in \eqref{pid2} we have
\begin{align*}
    \partial_tG_{i,T}^{a}(t,x,y,z) + \mathcal{L}^aG_{i,T}^{a}(t,x,y,z) = \ & \sum_{\substack{k\in \mathcal{J}\\ k\neq i}}\mu_{ik}(t)\left(G_{i,T}^{a}(t,x,y,z)-G_{k,T}^a(t,x,y,z)\right) \\
    &+\sum_{j\in\mathcal{J}}p_{ij}(t,T)\left(\partial_tU_T^{f_j,a}(t,x,y,z)+\mathcal{L}^aU_T^{f_j,a}(t,x,y,z)\right) \\
    = \ & \sum_{\substack{k\in \mathcal{J}\\ k\neq i}}\mu_{ik}(t)\left(G_{i,T}^{a}(t,x,y,z)-G_{k,T}^a(t,x,y,z)\right) \\
    &+r\sum_{j\in\mathcal{J}}p_{ij}(t,T)U_T^{f_j,a}(t,x,y,z) \\
    = \ & \sum_{\substack{k\in \mathcal{J}\\ k\neq i}}\mu_{ik}(t)\left(G_{i,T}^{a}(t,x,y,z)-G_{k,T}^a(t,x,y,z)\right) \\
    &+rG_{i,T}^{a}(t,x,y,z). 
\end{align*}
Replacing this last equality in \eqref{eq:prethiele} we obtain
\begin{align*}
    \partial_tV_i^a(t,x,y,z) = \ & rV_i^a(t,x,y,z)-g_i(t,x)-\sum_{\substack{k\in \mathcal{J}\\ k\neq i}}\mu_{ik}(t)\left(h_{ik}(t,x)+V_k^a(t,x,y,z)-V_i^a(t,x,y,z)\right)\\
    &-\mathcal{L}^aV_i^a(t,x,y,z),
\end{align*}
finishing the proof.
\end{proof}

\section*{\center Acknowledgments}

The authors would like to acknowledge financial support by the Research Council of Norway under the SCROLLER project, project number 299897.

\noindent

\appendix
\section{Appendix: Technical results}\label{sec: appendix}
We give all the proofs that were postponed.
\subsection{Moments of the variance and the Radon-Nikodym derivative}
\begin{lem}\label{Anmomentcompound}
Let $s\geq1$. Then, $\EE\left[L_t^s\right]<\infty$ for all $t\in[0,T]$ and $\int_0^T\EE\left[L_t^s\right]dt<\infty$.

Moreover, $\EE\left[[L]_t^s\right]<\infty$ for all $t\in[0,T]$.
\end{lem}
\begin{proof}
Applying Hölder's inequality for sums we get 
    \begin{align*}
        L_t^s\leq N_t^{s-1}\sum_{i=1}^{N_t}J_i^s.
    \end{align*}
    Using that $N$ and $\{J_i\}_{i\geq1}$ are independent we obtain 
    \begin{align*}
        \EE[L_t^s]&\leq\EE\left[N_t^{s-1}\sum_{i=1}^{N_t}J_i^s\right]=\EE\left[\EE\left[N_t^{s-1}\sum_{i=1}^{N_t}J_i^s\Big|\mathcal{F}_t^N\right]\right]=\EE\left[N_t^{s-1}\EE\left[\sum_{i=1}^{N_t}J_i^s\Big|\mathcal{F}_t^N\right]\right] \\
        &=\EE\left[N_t^{s-1}\sum_{i=1}^{N_t}\EE[J_i^s|\mathcal{F}_t^N]\right]=\EE\left[N_t^{s-1}\sum_{i=1}^{N_t}\EE[J_i^s]\right]=\EE\left[N_s^sJ_1^s\right]=\EE\left[N_s^s\right]\EE\left[J_1^s\right]<\infty,
    \end{align*}
    where we have used that $\EE[N_t^s]<\infty$ and $\EE[J_1^s]<\infty$. For a reference of $\EE[N_t^s]<\infty$ see \cite[Theorem 1]{momentshawkes} or \cite[Corollary 3.2]{CONTAGION} where their condition $\delta>\mu_{1_G}$ is our stability condition $\beta>\alpha$. For $\EE[J_1^s]<\infty$ see Assumption \ref{as}. Note that $\int_0^T\EE\left[L_t^s\right]dt<\infty$ is a just a consequence of $\EE[L_t^s]\leq\EE[L_T^s]<\infty$ for all $t\in[0,T]$

     Finally, in order to prove that $\EE\left[[L]_t^s\right]<\infty$ one can repeat the same argument with $[L]_t=\sum_{i=1}^{N_t}J_i^2$. 
    \end{proof}
 \begin{manuallemma}{\ref{lemintegr}}\label{lemAintegrnpower} Let $s\geq1$. Then, $\EE\left[v_t^s\right]<\infty$ for all $t\in[0,T]$ and $\int_0^T\EE\left[v_t^s\right]dt<\infty$. \end{manuallemma}
   \begin{proof}
       Recall that 
\begin{align*}
    v_t=v_0-\kappa\int_0^t(v_u-\Bar{v})du+\sigma\int_0^t\sqrt{v_u}dW_u+\eta L_u.
\end{align*}
Applying Hölder's inequality for sums we get
\begin{align}\label{vn}
    v_t^s\leq 4^{s-1}&\left[v_0^s+\kappa^s\left|\int_0^t(v_u-\Bar{v})du\right|^s+\sigma^s\left|\int_0^t\sqrt{v_u}dW_u\right|^s+\eta^sL_t^s\right].
\end{align}
By applying Jensen's inequality and again Hölder's inequality for sums we obtain
\begin{align*}
    \left|\int_0^t(v_u-\Bar{v})du\right|^s&\leq t^s\left(\frac{1}{t}\int_0^t|v_u-\Bar{v}|du\right)^s \\
    &\leq t^{s-1}\int_0^t|v_u-\Bar{v}|^sdu \\
    &\leq (2t)^{s-1}\int_0^t\left(v_u^s+\Bar{v}^s\right)du \\
    & \leq (2T)^{s-1}\int_0^Tv_u^sdu+2^{s-1}(T\Bar{v})^{s}.
\end{align*}
Replacing the last inequality in \eqref{vn} we get 
\begin{align*}
    v_t^s\leq 4^{s-1}\left[v_0^s+(2T)^{s-1}\kappa^s\int_0^Tv_u^sdu+2^{s-1}(T\Bar{v}\kappa)^s+\sigma^s\left|\int_0^t\sqrt{v_u}dW_u\right|^s+\eta^sL_T^s\right].
\end{align*}
Define now $A(t):=4^{s-1}\left[v_0^s+2^{s-1}(T\Bar{v}\kappa)^s+\sigma^s\left|\int_0^t\sqrt{v_u}dW_u\right|^s+\eta^sL_T^s\right]$, $B:=(8T)^{s-1}\kappa^s$. Then,
\begin{align*}
    v_t^s\leq A(t)+B\int_0^tv_u^sdu,
\end{align*}
for all $t\in[0,T]$. For each $\omega\in\Omega$, the functions $t\to v^s_t(\omega)$ and $t\to A(t,\omega)$ are measurable in $t$. Moreover, $\int_0^Tv_u^sdu<\infty$ a.s. because $v$ is a continuous process except a finite number of finite jumps. Note that $A(t)>0$ for all $t\in[0,T]$. Then, all the conditions to apply Grönwall's inequality hold and we get
\begin{align}\label{Agronwalln}
    v_t^s \leq A(t)+B\int_0^tA(u)e^{B(t-u)}du\leq A(t)+B e^{B T}\int_0^TA(u)du,
\end{align}
for all $t\in[0,T]$. 

By Proposition \ref{novikov}, there exists $c_l>0$ such that for $c<c_l$, $\EE\left[\exp\left(c\int_0^Tv_udu\right)\right]<\infty$. In particular, $\EE\left[\int_0^Tv_udu\right]<\infty$ and the process $t\to\int_0^t\sqrt{v_u}dW_u$ is a $(\mathcal{F},\PP)$-martingale. By the Burkholder–Davis–Gundy inequality, there exists a constant $C_s$ independent of the martingale and $t$ such that
\begin{align*}
    \EE\left[\left|\int_0^t\sqrt{v_u}dW_u\right|^s\right]\leq C_s\EE\left[\left(\int_0^tv_udu\right)^{\frac{s}{2}}\right]\leq C_s\EE\left[\left(\int_0^Tv_udu\right)^{\frac{s}{2}}\right]<\infty,
\end{align*}
where the last expectation is finite because all positive moments of $\int_0^Tv_udu$ exist. By Lemma \ref{Anmomentcompound} we have that $\EE[L_T^s]<\infty$ and, therefore
\begin{align*}
    \EE\left[A(t)\right]\leq 4^{s-1}\left[v_0^s+2^{s-1}(T\Bar{v}\kappa)^s+\sigma^sC_s\EE\left[\left(\int_0^Tv_udu\right)^{\frac{s}{2}}\right]+\eta^s\EE[L_T^s]\right]=:E<\infty,
\end{align*}
for all $t\in[0,T]$. Finally, by taking expectations in \eqref{Agronwalln} we obtain 
\begin{align*}
    \EE[v_t^s]\leq E+B e^{B T}TE<\infty,
\end{align*}
for all $t\in[0,T]$ and it follows that $\int_0^T\EE[v_t^s]dt<\infty$. 
   \end{proof}
\begin{definition}\label{def1f1}
For $a,b,z\in\RR$, we define the confluent hypergeometric function of the first kind ${}_1F_1(a,b;z)$ in the following way
\begin{align*}
 {}_1F_1(a,b;z):=\sum_{n=0}^\infty\frac{a^{(n)}}{b^{(n)}}\frac{z^n}{n!}
\end{align*}
where $q^{(0)}=1$ and $q^{(n)}=q\cdot (q+1)\cdot ... \cdot (q+n-1)$ for $n\geq1$ is the rising factorial. See \cite[Section 5.3]{paolella_2007} and \cite[Section 13]{asimptotic1f1} for more details about the confluent hypergeometric function of the first kind. 
\end{definition}
    \begin{manuallemma}{\ref{leminver}}\label{lemAinver} Let $s\geq1$. If $2\kappa \Bar{v}>s\sigma^2$, $\EE\left[\frac{1}{v_t^s}\right]<\infty$ for all $t\in[0,T]$ and $\int_0^T\EE\left[\frac{1}{v_t^s}\right]dt<\infty$. 
   \end{manuallemma}
\begin{proof}
    Recall that $\widetilde{v}=\{\widetilde{v}_t, t\in[0,T]\}$ is the pathwise unique strong solution of 
\begin{align*}
    \widetilde{v}_t=\widetilde{v}_0-\kappa\int_0^t\left(\widetilde{v}_s-\Bar{v}\right)ds+\sigma\int_0^t\sqrt{\widetilde{v}_s}dW_s.
\end{align*}    
    In \cite[Section 2, Equation (2.3)]{inverseheston}, we see that for $t\in (0,T]$
\begin{align*}
    \widetilde{v}_t \sim \frac{e^{-\kappa t}v_0}{k(t)}\chi_\delta^{'2}\left(k(t)\right) \ \ \text{with} \ \ k(t):=\frac{4\kappa v_0 e^{-\kappa t}}{\sigma^2(1-e^{-\kappa t})} \ \ \text{and} \ \ \delta:=\frac{4\kappa\Bar{v}}{\sigma^2},
\end{align*}
where $\chi_\delta^{'2}(k(t))$ denotes a noncentral chi-square random variable with $\delta$ degrees of freedom and noncentrality parameter $k(t)$. Since $2\kappa\Bar{v}>s\sigma^2$, $-s>-\frac{2\kappa\Bar{v}}{\sigma^2}=-\frac{\delta}{2}$. Then, we can use \cite[Section 10.1, Equation (10.9)]{paolella_2007} to get that for $t>0$
\begin{align*}
    \EE\left[\left(\frac{1}{\widetilde{v}_t}\right)^s\right]=\left(\frac{k(t)}{e^{-\kappa t}v_0}\right)^s\frac{1}{2^se^{k(t)/2}}\frac{\Gamma\left(\frac{\delta}{2}-s\right)}{\Gamma\left(\frac{\delta}{2}\right)}{}_1F_1\left(\frac{\delta}{2}-s,\frac{\delta}{2};\frac{k(t)}{2}\right)<\infty,
\end{align*}
where ${}_1F_1$ is the confluent hypergeometric function of the first kind given in Definition \ref{def1f1}. By Proposition \ref{p1}, this proves that we have $\EE\left[\frac{1}{v_t^s}\right]<\infty$ for all $t\in[0,T]$. 

To prove that the $s$th negative moment is integrable we check that $t\to \EE\left[\left(\frac{1}{\widetilde{v}_t}\right)^s\right]$ is a continuous function for $t\in[0,T]$. Since $\frac{\delta}{2}$ is strictly positive, and $k\colon(0,T]\to\RR$ is continuous, the function $t\to {}_1F_1\left(\frac{\delta}{2}-s,\frac{\delta}{2};\frac{k(t)}{2}\right)$ is continuous for $t\in(0,T]$. Therefore, $t\to \EE\left[\left(\frac{1}{\widetilde{v}_t}\right)^s\right]$ is continuous at least for $t\in(0,T]$. To check continuity at $t=0$ observe that $\lim_{t\rightarrow0^+}k(t)=\infty$ and by \cite[Section 13, Equation 13.1.4]{asimptotic1f1} it is known that
\begin{align*}
    \lim_{t\rightarrow0^+}\frac{{}_1F_1\left(\frac{\delta}{2}-s,\frac{\delta}{2};\frac{k(t)}{2}\right)}{\frac{\Gamma\left(\frac{\delta}{2}\right)}{\Gamma\left(\frac{\delta}{2}-s\right)}e^{k(t)/2}\left(\frac{k(t)}{2}\right)^{-s}}=1.
\end{align*}
Then, 
\begin{align*}
    \lim_{t\rightarrow0^+}\EE\left[\left(\frac{1}{\widetilde{v}_t}\right)^s\right]&=\lim_{t\rightarrow0^+}\left(\frac{k(t)}{e^{-\kappa t}v_0}\right)^s\frac{1}{2^se^{k(t)/2}}\frac{\Gamma\left(\frac{\delta}{2}-s\right)}{\Gamma\left(\frac{\delta}{2}\right)} {}_1F_1\left(\frac{\delta}{2}-s,\frac{\delta}{2};\frac{k(t)}{2}\right) \\
    & = \lim_{t\rightarrow0^+} \left(\frac{k(t)}{e^{-\kappa t}v_0}\right)^s\frac{1}{2^se^{k(t)/2}}\frac{\Gamma\left(\frac{\delta}{2}-s\right)}{\Gamma\left(\frac{\delta}{2}\right)}\frac{\Gamma\left(\frac{\delta}{2}\right)}{\Gamma\left(\frac{\delta}{2}-s\right)}e^{k(t)/2}\left(\frac{k(t)}{2}\right)^{-s} \\
    & = \frac{1}{v_0^s}=\EE\left[\left(\frac{1}{\widetilde{v}_0}\right)^s\right].
\end{align*}
This proves that $t\to \EE\left[\left(\frac{1}{\widetilde{v}_t}\right)^s\right]$ is a continuous function for $t\in[0,T]$. Thus, $\int_0^T\EE\left[\left(\frac{1}{\widetilde{v}_t}\right)^s\right]dt<\infty$, which by Proposition \ref{p1} implies that $\int_0^T\EE\left[\frac{1}{v_t^s}\right]dt<\infty$. 
\end{proof}

\begin{lem}\label{lemAinvexp}
If $2\kappa\Bar{v}>\sigma^2$ and $c\leq\frac{1}{2}\left(\frac{2\kappa\Bar{v}-\sigma^2}{2\sigma}\right)^2$, then
\begin{align*}
    \EE\left[\exp\left(c\int_0^T\frac{1}{v_u}du\right)\right]<\infty.
\end{align*}
\end{lem}
\begin{proof}
We define the process $Z=\{Z_t=\frac{1}{\widetilde{v}_t}, t\in[0,T]\}$ where $\widetilde{v}=\{\widetilde{v}_t, t\in[0,T]\}$ is the pathwise unique strong solution of 
\begin{align*}
    \widetilde{v}_t=\widetilde{v}_0-\kappa\int_0^t\left(\widetilde{v}_u-\Bar{v}\right)du+\sigma\int_0^t\sqrt{\widetilde{v}_u}dW_u.
\end{align*}
Note that the Feller condition $2\kappa\Bar{v}\geq\sigma^2$ ensures that the process $\widetilde{v}$ is strictly positive. Therefore, the process $Z$ is well defined. Applying Itô formula we get
\begin{align}\label{invheston}
    dZ_t=\left[\kappa Z_t+\left(\sigma^2-\kappa\Bar{v}\right)Z_t^2\right]dt-\sigma Z_t^{3/2}dW_t.
\end{align}
Following the notation in \cite[Theorem 3]{3over2}, $Z$ is a quadratic drift $3/2$ process with parameters $p(t)=\kappa$, $q=\sigma^2-\kappa\Bar{v}$ and $\epsilon=-\sigma$. The condition $q<\frac{\epsilon^2}{2}$ is satisfied because $2\kappa\Bar{v}>\sigma^2$. Applying \cite[Theorem 3]{3over2} with $u=0$ and $s=-c$ we get that for $c\leq\frac{1}{2}\left(\frac{2\kappa\Bar{v}-\sigma^2}{2\sigma}\right)^2$
\begin{align*}
    \EE\left[\exp\left(c\int_0^T\frac{1}{\widetilde{v}_u}du\right)\right]=\frac{\Gamma(\gamma-\widetilde{\alpha})}{\Gamma(\widetilde{\alpha})}\left(\frac{2}{\sigma^2y(0,1/v_0)}\right)^{\widetilde{\alpha}} {}_1F_1\left(\widetilde{\alpha},\gamma;-\frac{2}{\sigma^2y(0,1/v_0)}\right)<\infty,
\end{align*}
where
\begin{align*}
    \widetilde{\alpha} & = -\left(\frac{\kappa\Bar{v}}{\sigma^2}-\frac{1}{2}\right)+\sqrt{\left(\frac{\kappa\Bar{v}}{\sigma^2}-\frac{1}{2}\right)^2-2\frac{c}{\sigma^2}}, \\
    \gamma & = 2\left(\widetilde{\alpha}+\frac{k\Bar{v}}{\sigma^2}\right), \\
    y(0,1/v_0) & =\frac{e^{\kappa T}-1}{v_0\kappa}.
\end{align*}
and ${}_1F_1$ is the confluent hypergeometric function defined in Definition \ref{def1f1}. Note that $\widetilde{\alpha}\in\RR$ because $c\leq\frac{1}{2}\left(\frac{2\kappa\Bar{v}-\sigma^2}{2\sigma}\right)^2$ by hypothesis. Finally, by Proposition \ref{p1}, we conclude that 
\begin{align*}
    \EE\left[\exp\left(c\int_0^T\frac{1}{v_u}du\right)\right]<\infty.
\end{align*}
See also \cite[Footnote 10 on page 136]{frontiers} for another reference about the inverse CIR. Note that \cite[Equation (24) on page 137]{frontiers} is the same condition we have on $c$.
\end{proof}
\begin{manuallemma}{\ref{nmoment}}\label{Anmoment} Let $\QQ(a)\in\mathcal{E}$, $s>1$, $D:=\sup_{t\in[0,T]}(\mu_t-r)^2<\infty$ and $X^{(a)}$ defined in Theorem \ref{risk}. Assume that $2\kappa\Bar{v}>\sigma^2$ and $\frac{1-\rho^2}{D(s^2-s)}\left(\frac{2\kappa\Bar{v}-\sigma^2}{2\sigma}\right)^2>1$. Consider $q_2$ such that $1<q_2<\frac{1-\rho^2}{D(s^2-s)}\left(\frac{2\kappa\Bar{v}-\sigma^2}{2\sigma}\right)^2$ and define $q_1:=\frac{q_2}{q_2-1}>1$. If
\begin{align}\label{Aconda}
    |a|<\min\Bigg\{\frac{1}{q_1s}\sqrt{\frac{c_l}{2}},\sqrt{\frac{(1-\rho^2)c_l}{q_1s\left[2q_1s(1-\rho^2)+\rho^2s-1\right]}}\Bigg\},
\end{align}
then
\begin{align*}
    \EE\left[\left(X_T^{(a)}\right)^s\right]<\infty \hspace{0.5cm}\text{and}\hspace{0.5cm} \EE\left[\left(X_t^{(a)}\right)^s\right]\leq \left(\frac{s}{s-1}\right)^s\EE\left[\left(X_T^{(a)}\right)^s\right]<\infty,
\end{align*}
for all $t\in[0,T]$.
\end{manuallemma}
\begin{proof}
Recall that we have defined $\theta_t^{(a)}:=\frac{1}{\sqrt{1-\rho^2}}\left(\frac{\mu_t-r}{\sqrt{v_t}}-a\rho\sqrt{v_t}\right)$, 
\begin{align*}
    Y_t^{(a)} &:=\exp\left(-\int_0^t\theta_u^{(a)}dB_u-\frac{1}{2}\int_0^t(\theta_u^{(a)})^2du\right), \\
    Z_t^{(a)} &:=\exp\left(-a\int_0^t\sqrt{v_u}dW_u-\frac{1}{2}a^2\int_0^tv_udu\right)
\end{align*}
and $X_t^{(a)}:=Y_t^{(a)}Z_t^{(a)}$ in Theorem \ref{risk}. By Proposition \ref{p1}, the variance process $v$ is strictly positive. This implies that $\int_0^T(\theta_u^{(a)})^2ds<\infty$, $\PP$-a.s, and since  $\theta^{(a)}$ is $\{\mathcal{F}_t^{W}\vee\mathcal{F}_t^L\}_{t\in[0,T]}$-adapted,
\begin{align*}
    Y^{(a)}_T|\mathcal{F}_T^{W}\vee\mathcal{F}_T^L\sim\text{Lognormal}\left(-\frac{1}{2}\int_0^T(\theta_u^{(a)})^2du,\int_0^T(\theta_u^{(a)})^2du\right).
\end{align*}
Using that $Z_T^{(a)}$ is $\mathcal{F}_T^{W}\vee\mathcal{F}_T^L$-measurable we have
\begin{align*}
    \EE\left[\left(X_T^{(a)}\right)^s\right] & =\EE\left[\left(Z_T^{(a)}\right)^s\EE\left[\left(Y_T^{(a)}\right)^s|\mathcal{F}_T^{W}\vee\mathcal{F}_T^L\right]\right] \\
    & = \EE\left[\left(Z_T^{(a)}\right)^s\exp\left(\left(\frac{s^2-s}{2}\right)\int_0^T(\theta_u^{(a)})^2du\right)\right].
\end{align*}
Using that $(\theta_u^{(a)})^2=\frac{1}{1-\rho^2}\left(\frac{(\mu_u-r)^2}{v_u}+a^2\rho^2v_u-2a\rho(\mu_u-r)\right)$ we obtain 
\begin{align*}
   \EE\left[\left(X_T^{(a)}\right)^s\right]=e^{O^{(a)}\int_0^T(\mu_u-r)du}\EE\left[\exp\left(P^{(a)}\int_0^T\sqrt{v_u}dW_u+Q^{(a)}\int_0^Tv_udu+R\int_0^T\frac{(\mu_u-r)^2}{v_u}du\right)\right],
\end{align*}
where
\begin{align}\label{C}
O^{(a)} & := -\frac{(s^2-s)a\rho}{1-\rho^2} \notag\\ 
P^{(a)} & := -as \notag\\ 
    Q^{(a)} & :=\frac{(s^2-s)a^2\rho^2}{2(1-\rho^2)}-\frac{1}{2}sa^2=\frac{sa^2(\rho^2s-1)}{2(1-\rho^2)} \\
    R & := \frac{s^2-s}{2(1-\rho^2)} \label{D},
\end{align}
We apply Hölder's inequality with $q_1=\frac{q_2}{q_2-1}>1$ and $q_2>1$ to get
\begin{align}\label{in1}
    & \EE\left[\exp\left(P^{(a)}\int_0^T\sqrt{v_u}dW_u+Q^{(a)}\int_0^Tv_udu+R\int_0^T\frac{(\mu_u-r)^2}{v_u}du\right)\right] \notag\\
    &\leq \EE\left[\exp\left(q_1 P^{(a)}\int_0^T\sqrt{v_u}dW_u+q_1Q^{(a)}\int_0^Tv_udu\right)\right]^\frac{1}{q_1}\EE\left[\exp\left(q_2 R\int_0^T\frac{(\mu_u-r)^2}{v_u}du\right)\right]^\frac{1}{q_2}.
\end{align}
Now, we focus on the first term in \eqref{in1}. We add and subtract the constant $q_1^2(P^{(a)})^2$ and we apply Cauchy-Schwarz inequality
\begin{align}\label{in}
    &\EE\left[\exp\left(q_1 P^{(a)}\int_0^T\sqrt{v_u}dW_u+q_1Q^{(a)}\int_0^Tv_udu\right)\right]  \notag\\
    & = \EE\left[\exp\left(q_1 P^{(a)}\int_0^T\sqrt{v_u}dW_u-q_1^2(P^{(a)})^2\int_0^Tv_udu\right)\exp\left((q_1 Q^{(a)}+q_1^2(P^{(a)})^2)\int_0^Tv_udu\right)\right] \notag\\
    & \leq \EE\left[\exp\left(2q_1P^{(a)}\int_0^T\sqrt{v_u}dW_u-2q_1^2(P^{(a)})^2\int_0^Tv_udu\right)\right]^\frac{1}{2}\EE\left[\exp\left(2(q_1 Q^{(a)}+q_1^2(P^{(a)})^2)\int_0^Tv_udu\right)\right]^\frac{1}{2}.
\end{align}
Note that the first term in \eqref{in} is the expectation of a Doléans-Dade exponential. 
By \eqref{Aconda}, $|a|<\frac{1}{q_1s}\sqrt{\frac{c_l}{2}}$ and we have that $2q_1^2(P^{(a)})^2=2q_1^2a^2s^2<c_l$. Then, by Proposition \ref{novikov}, Novikov's condition is satisfied, that is
\begin{align*}
   \EE\left[\exp\left(2q_1^2(P^{(a)})^2\int_0^Tv_udu\right)\right]<\infty.
\end{align*}
Therefore, 
\begin{align*}
    \EE\left[\exp\left(2q_1P^{(a)}\int_0^T\sqrt{v_u}dW_u-2q_1^2(P^{(a)})^2\int_0^Tv_udu\right)\right]<\infty. 
\end{align*}
For the second term in \eqref{in} we need to check again that Proposition \ref{novikov} is satisfied. One can check that
\begin{align*}
    2(q_1 Q^{(a)}+q_1^2(P^{(a)})^2)=\frac{q_1sa^2}{1-\rho^2}\left[2q_1s(1-\rho^2)+\rho^2s-1\right].
\end{align*}
By Observation \ref{obspositivity}, $2q_1s(1-\rho^2)+\rho^2s-1>0$ and by \eqref{Aconda}, $|a|<\sqrt{\frac{(1-\rho^2)c_l}{q_1s\left[2q_1s(1-\rho^2)+\rho^2s-1\right]}}$. Thus, $2(q_1 Q^{(a)}+q_1^2(P^{(a)}))^2<c_l$ and applying again Proposition \ref{novikov} we obtain
\begin{align*}
 \EE\left[\exp\left( 2(q_1 Q^{(a)}+q_1^2(P^{(a)}))^2\int_0^Tv_udu\right)\right]<\infty.
 \end{align*}
We conclude that the two terms in \eqref{in} are finite and, therefore, the first expectation in \eqref{in1} is finite as well, that is, 
\begin{align*}
    \EE\left[\exp\left(q_1 P^{(a)}\int_0^T\sqrt{v_u}dW_u+q_1 Q^{(a)}\int_0^Tv_udu\right)\right] <\infty.
\end{align*}
We check the second term in \eqref{in1}. Recall that $D=\sup_{u\in[0,T]}(\mu_u-r)^2<\infty$. Then
\begin{align*}
    \EE\left[\exp\left(q_2 R\int_0^T\frac{(\mu_u-r)^2}{v_u}du\right)\right] \leq \EE\left[\exp\left(q_2 RD\int_0^T\frac{1}{v_u}du\right)\right].
\end{align*}
Since $2\kappa\Bar{v}>\sigma^2$ and $q_2$ is such that $1<q_2<\frac{1-\rho^2}{D(s^2-s)}\left(\frac{2\kappa\Bar{v}-\sigma^2}{2\sigma}\right)^2$ and $q_2RD=\frac{q_2(s^2-s)D}{2(1-\rho^2)}$, we can apply Lemma \ref{lemAinvexp} to obtain
\begin{align*}
     \EE\left[\exp\left(q_2 R\int_0^T\frac{(\mu_u-r)^2}{v_u}du\right)\right]<\infty.
\end{align*}
This proves that the second term in \eqref{in1} is also finite and we can conclude that  $\EE\left[\left(X_T^{(a)}\right)^s\right]<\infty$. 

\bigskip
\noindent
By Theorem \ref{risk}, $X^{(a)}$ is a positive $(\mathcal{F},\PP)$-martingale. We can apply Doob's martingale inequality, see \cite[Section 2.1.2, Theorem 2.1.5]{applebaum_2009}, to get
\begin{align*}
    \EE\left[\left(X_t^{(a)}\right)^s\right]\leq\EE\left[\sup_{0\leq u\leq T}\left(X_u^{(a)}\right)^s\right]\leq\left(\frac{s}{s-1}\right)^s\EE\left[\left(X_T^{(a)}\right)^s\right]<\infty
\end{align*}
for any $t\in[0,T]$. This finishes the proof of this lemma. 
\end{proof}
\subsection{Compensator of $N$ under the historic and the risk neutral measures}
\begin{lem}\label{lemmarting}
    Let $\QQ(a)\in\mathcal{E}_{m}(Q_1,2+\varepsilon_1)$, the following holds
    \begin{enumerate}
        \item The process $t\to\int_0^tX_u^{(a)}d(L-\Lambda^L)_u$ is a square integrable $(\mathcal{F},\PP)$-martingale.
        \item The process $t\to\int_0^tL_{u-}dX_u^{(a)}$ is a square integrable $(\mathcal{F},\PP)$-martingale.
        \item Let $0\leq s \leq t\leq T$, then $\EE[L_tX_t^{(a)}|\mathcal{F}_s]=L_sX_s^{(a)}+\EE[J_1]\int_s^t\EE[\lambda_uX_u^{(a)}|\mathcal{F}_s]du$.
    \end{enumerate}
\end{lem}
\begin{proof}
    (1) Recall that we have defined $\theta_t^{(a)}:=\frac{1}{\sqrt{1-\rho^2}}\left(\frac{\mu_t-r}{\sqrt{v_t}}-a\rho\sqrt{v_t}\right)$, 
\begin{align*}
    Y_t^{(a)} &:=\exp\left(-\int_0^t\theta_u^{(a)}dB_u-\frac{1}{2}\int_0^t(\theta_u^{(a)})^2du\right), \\
    Z_t^{(a)} &:=\exp\left(-a\int_0^t\sqrt{v_u}dW_u-\frac{1}{2}a^2\int_0^tv_udu\right)
\end{align*}
and $X_t^{(a)}:=Y_t^{(a)}Z_t^{(a)}$ in Theorem \ref{risk}. By Lemma \ref{mainlemcomp}, $L-\Lambda^L$ is a square integrable $(\mathcal{F},\PP)$-martingale. To prove that $t\to\int_0^tX_u^{(a)}d(L-\Lambda^L)_u$ is a square integrable $(\mathcal{F},\PP)$-martingale we check that 
    \begin{align*}
     \EE\left[\int_0^T\left(X_u^{(a)}\right)^2d[L-\Lambda^L]_u\right]<\infty. 
    \end{align*}
The quadratic variation of the compound Hawkes process is given by $[L]_t=\sum_{i=1}^{N_t}J_i^2$ and $[L-\Lambda^L]_t=[L]_t$. Using Hölder's inequality with $p=1+\frac{\varepsilon_1}{2}>1$ and $q=\frac{p}{p-1}>1$ (recall Assumption \ref{as3} and Observation \ref{momentsobs} for the definition of $\varepsilon_1$) and  Doob's martingale inequality, see \cite[Section 2.1.2, Theorem 2.1.5]{applebaum_2009}, in the last step we have
\begin{align*}
    \EE\left[\int_0^T\left(X_u^{(a)}\right)^2d[L-\Lambda^L]_u\right] &= \EE\left[\int_0^T\left(X_u^{(a)}\right)^2d[L]_u\right] \\
    & \leq \EE\left[\sup_{t\in[0,T]}\left(X_t^{(a)}\right)^2[L]_T\right] \\
    & \leq \EE\left[\sup_{t\in[0,T]}\left(X_t^{(a)}\right)^{2p}\right]^{1/p}\EE\left[[L]_T^q\right]^{1/q}  \\
    & \leq \EE\left[\sup_{t\in[0,T]}\left(X_t^{(a)}\right)^{2+\varepsilon_1}\right]^{1/p}\EE\left[[L]_T^q\right]^{1/q}  \\
    & \leq q\EE\left[\left(X_T^{(a)}\right)^{2+\varepsilon_1}\right]^{1/p}\EE\left[[L]_T^q\right]^{1/q} <\infty,
\end{align*}
where $\EE\left[\left(X_T^{(a)}\right)^{2+\varepsilon_1}\right]<\infty$ by Observation \ref{momentsobs} and $\EE\left[[L]_T^q\right]<\infty$ by Lemma \ref{Anmomentcompound}. We conclude that the process $t\to\int_0^tX_u^{(a)}d(L-\Lambda^L)_u$ is a square integrable $(\mathcal{F},\PP)$-martingale.

(2)  By Theorem \ref{risk} and Observation \ref{momentsobs}, $X^{(a)}$ is a square integrable martingale. To prove that $t\to\int_0^tL_{u-}dX_u^{(a)}$ is a square integrable $(\mathcal{F},\PP)$-martingale we check that
\begin{align*}
    \EE\left[\int_0^TL_{u-}^2d[X^{(a)}]_u\right]<\infty.
\end{align*}
The quadratic variation of $X^{(a)}$ is given by $d[X^{(a)}]_t=\left[\left(\theta_t^{(a)}\right)^2+a^2v_t\right]\left(X_t^{(a)}\right)^2dt$. Then, 
\begin{align}\label{co}
    \EE\left[\int_0^TL_{u-}^2d[X^{(a)}]_u\right] & = \EE\left[\int_0^TL_{u-}^2\left[\left(\theta_u^{(a)}\right)^2+a^2v_u\right]\left(X_u^{(a)}\right)^2du\right]\notag \\ 
    & =\int_0^T\EE\left[L_{u-}^2\left(\theta_u^{(a)}\right)^2\left(X_u^{(a)}\right)^2\right]du+a^2\int_0^T\EE\left[L_{u-}^2v_u\left(X_u^{(a)}\right)^2\right]du.
\end{align}
We focus on the first term in \eqref{co}, applying Hölder's inequality with $p_1=\frac{p_2p_3}{p_2p_3-p_2-p_3}>1$, $p_2=1+\varepsilon_2>1$, $p_3=1+\frac{\varepsilon_1}{2}>1$,  and then Doob's martingale inequality, see \cite[Section 2.1.2, Theorem 2.1.5]{applebaum_2009}, to the last expectation we get
\begin{align}\label{cota}
    \EE\left[L_{u-}^2\left(\theta_u^{(a)}\right)^2\left(X_u^{(a)}\right)^2\right]&\leq\EE[L_u^{2p_1}]^{\frac{1}{p_1}}\EE\left[\left(\theta_u^{(a)}\right)^{2p_2}\right]^{\frac{1}{p_2}}\EE\left[\left(X_u^{(a)}\right)^{2p_3}\right]^{\frac{1}{p_3}} \notag\\
    & = \EE[L_u^{2p_1}]^{\frac{1}{p_1}}\EE\left[\left(\theta_u^{(a)}\right)^{2+2\varepsilon_2}\right]^{\frac{1}{p_2}}\EE\left[\left(X_u^{(a)}\right)^{2+\varepsilon_1}\right]^{\frac{1}{p_3}} \notag\\
    & \leq \left(\frac{2+\varepsilon_1}{1+\varepsilon_1}\right)^{2}\EE[L_T^{2p_1}]^{\frac{1}{p_1}}\EE\left[\left(\theta_u^{(a)}\right)^{2+2\varepsilon_2}\right]^{\frac{1}{p_2}}\EE\left[\left(X_T^{(a)}\right)^{2+\varepsilon_1}\right]^{\frac{1}{p_3}}.
\end{align}
Note that $\EE[L_T^{2p_1}]<\infty$ by Lemma \ref{Anmomentcompound} and $\EE\left[\left(X_T^{(a)}\right)^{2+\varepsilon_1}\right]<\infty$ by Observation \ref{momentsobs}. Applying Hölder's inequality for sums we get
\begin{align*}
    \EE\left[\left(\theta_u^{(a)}\right)^{2+2\varepsilon_2}\right]\leq\frac{2^{1+2\varepsilon_2}}{(1-\rho^2)^{1+\varepsilon_2}}\left[D^{1+\varepsilon_2}\EE\left[\left(\frac{1}{v_u}\right)^{1+\varepsilon_2}\right]+(a\rho)^{2+2\varepsilon_2}\EE\left[v_u^{1+\varepsilon_2}\right]\right]
\end{align*}
where  $D=\sup_{t\in[0,T]}(\mu_t-r)^2<\infty$. By Observation \ref{momentsobs} and Lemma \ref{lemintegr} we have that
\begin{align*}
    \int_0^T\EE\left[\left(\theta_u^{(a)}\right)^{2+2\varepsilon_2}\right]<\infty.
\end{align*}
Applying Hölder's inequality with $p_2>1$ and $q_2=\frac{p_2}{p_2-1}>1$ we obtain
\begin{align*}
    \int_0^T\EE\left[\left(\theta_u^{(a)}\right)^{2+2\varepsilon_2}\right]^{\frac{1}{p_2}}du\leq T^{\frac{p_2-1}{p_2}}\left(\int_0^T\EE\left[\left(\theta_u^{(a)}\right)^{2+2\varepsilon_2}\right]du\right)^{\frac{1}{p_2}}<\infty.
\end{align*}
By \eqref{cota} this implies that
\begin{align*}
    \int_0^T\EE\left[L_{u-}^2\left(\theta_u^{(a)}\right)^2\left(X_u^{(a)}\right)^2\right]du<\infty.
\end{align*}
Similarly, using Lemma \ref{Anmomentcompound}, Lemma \ref{lemAintegrnpower} and Observation \ref{momentsobs} we can show that the second term in \eqref{co} is finite. We conclude that the process $t\to\int_0^tL_{u-}dX_u^{(a)}$ is a square integrable $(\mathcal{F},\PP)$-martingale.

(3) Applying Itô formula, using that $L_0=0$, that $L$ is of finite variation and $X^{(a)}$ is continuous we have
\begin{align*}
    L_tX_t^{(a)}=\int_0^tL_{u-}dX_u^{(a)}+\int_0^tX_u^{(a)}dL_u.
\end{align*}
Using that the processes  $t\to\int_0^tL_{u-}dX_u^{(a)}$ and  $t\to\int_0^tX_u^{(a)}d(L-\Lambda^L)_u$ are square integrable  $(\mathcal{F},\PP)$-martingales and the expression of $\Lambda^L$ given in Lemma \ref{mainlemcomp} we obtain
\begin{align*}
\EE[L_tX_t^{(a)}|\mathcal{F}_s]&=\EE\left[\int_0^tL_{u-}dX_u^{(a)}\Big|\mathcal{F}_s\right]+\EE\left[\int_0^tX_u^{(a)}dL_u\Big|\mathcal{F}_s\right] \\
    & = \int_0^sL_{u-}dX_u^{(a)}+\int_0^sX_u^{(a)}dL_u+\EE\left[\int_s^tX_u^{(a)}dL_u\Big|\mathcal{F}_s\right]\\
     & = L_sX_s^{(a)}+\EE\left[\int_s^tX_u^{(a)}d(L-\Lambda^L)_u\Big|\mathcal{F}_s\right]+\EE\left[\int_s^tX_u^{(a)}d\Lambda^L_u\Big|\mathcal{F}_s\right] \\
     & = L_sX_s^{(a)}+\EE[J_1]\int_s^t\EE[\lambda_uX_u^{(a)}|\mathcal{F}_s]du.
\end{align*}
\end{proof}

\begin{manualprop}{\ref{qacomhawkesmain}}\label{qacomhawkes} Let $\QQ(a)\in\mathcal{E}_{m}(Q_1,2+\varepsilon_1)$, then
\begin{enumerate}
    \item $N-\Lambda^N$ is a $(\mathcal{F},\QQ(a))$-martingale.
    \item $L-\Lambda^L$ is a $(\mathcal{F},\QQ(a))$-martingale.
\end{enumerate}
\end{manualprop}
\begin{proof}
First of all, note that for $0\leq s\leq t\leq T$, 
\begin{align*}
\EE\left[\int_0^t\lambda_uX_t^{(a)}du\Big|\mathcal{F}_s\right]&=\int_0^s\EE[\lambda_uX_t^{(a)}|\mathcal{F}_s]du+\int_s^t\EE[\lambda_uX_t^{(a)}|\mathcal{F}_s]du \\
& = X_s^{(a)}\int_0^s\lambda_udu+\int_s^t\EE[\EE[\lambda_uX_t^{(a)}|\mathcal{F}_u]\mathcal{F}_s]du \\
& = X_s^{(a)}\int_0^s\lambda_udu+\int_s^t\EE[\lambda_uX_u^{(a)}|\mathcal{F}_s]du.
\end{align*}
    By Lemma \ref{lemmarting} we know that $\EE[L_tX_t^{(a)}|\mathcal{F}_s]=L_sX_s^{(a)}+\EE[J_1]\int_s^t\EE[\lambda_uX_u^{(a)}|\mathcal{F}_s]du$. Using the previous equalities we obtain, 
    \begin{align*}
        \EE^{\QQ(a)}\left[L_t-\EE[J_1]\int_0^t\lambda_udu \Big| \mathcal{F}_s\right]  = & \ \frac{1}{X_s^{(a)}}\EE\left[L_tX_t^{(a)}-\EE[J_1]\int_0^t\lambda_uX_t^{(a)}du\Big|\mathcal{F}_s\right] \\
          = \ &  \frac{1}{X_s^{(a)}}\Bigg[L_sX_s^{(a)}+\EE[J_1]\int_s^t\EE[\lambda_uX_u^{(a)}|\mathcal{F}_s]du
          \\ 
          & -\EE[J_1]X_s^{(a)}\int_0^s\lambda_udu-\EE[J_1]\int_s^t\EE[\lambda_uX_u^{(a)}|\mathcal{F}_s]du\Bigg] \\
          = \ & L_s-\EE[J_1]\int_0^s\lambda_udu.
    \end{align*}
    This finishes the proof.
\end{proof}

\subsection{Derivation of Thiele's PIDE for unit-linked policies}
\begin{manuallemma}{\ref{lem2}}\label{lem2A} Let $\QQ(a)\in\mathcal{E}_{m}(Q_1,2+\varepsilon_1)$, $\varphi\colon[0,T]\times\RR_+\to\RR_+$ such that $\EE^{\QQ(a)}[|\varphi(s,S_s)|]<\infty$ for all $s\in[0,T]$. Then, there exists a function $U^{\varphi,a}\colon[0,T]^2\times\mathcal{D}\to\RR_+$ such that
\begin{align}\label{Arelation}
    e^{-r(s-t)}\EE^{\QQ(a)}[\varphi(s,S_s) | \mathcal{F}_t ]=U_s^{\varphi,a}(t,S_t,v_t,\lambda_t)
\end{align}
where $s,t\in[0,T]$. Note that $U_s^{\varphi,a}(t,x,y,z)=e^{-r(s-t)}\varphi(s,x)$ for $t\in[s,T]$. 
\newline
Furthermore, fix  $s\in[0,T]$, if $U_s^{\varphi,a}\in\mathcal{C}^{1,2}$, it satisfies the following PIDE
\begin{align}\label{Apid2}
    \partial_t U_s^{\varphi,a}(t,x,y,z)+\mathcal{L}^aU_s^{\varphi,a}(t,x,y,z)=rU_s^{\varphi,a}(t,x,y,z),
\end{align}
where $\mathcal{L}^a$ is defined in Definition \ref{loperator}, $(t,x,y,z)\in[0,s]\times\mathcal{D}$ and final condition $U_s^{\varphi,a}(s,x,y,z)=\varphi(s,x)$.
\end{manuallemma}
\begin{proof}
    By Lemma \ref{lema1} there exists a function $Z^{\varphi,a}\colon[0,T]^2\times\mathcal{D}\to\RR_+$  such that
    \begin{align*}
        \EE^{\QQ(a)}[\varphi(s,S_s) | \mathcal{F}_t ]=Z_{s}^{\varphi,a}(t,S_t,v_t,\lambda_t),
    \end{align*}
where $s,t\in[0,T]$. Define the function $U^{\varphi,a}\colon[0,T]^2\times\mathcal{D}\to\RR_+$ by 
\begin{align*}
    U_s^{\varphi,a}(t,x,y,z):=e^{-r(s-t)}Z_{s}^{\varphi,a}(t,x,y,z).
\end{align*}
Then, \eqref{Arelation} is satisfied. Moreover, fix $s\in[0,T]$, if $U_s^{\varphi,a}\in\mathcal{C}^{1,2}$, $Z_{s}^{\varphi,a}\in\mathcal{C}^{1,2}$ and it satisfies the following PIDE 
    \begin{align}\label{Apidbefore}
         \partial_tZ^{\varphi,a}_s(t,x,y,z)+\mathcal{L}^aZ^{\varphi,a}_s(t,x,y,z)=0.
    \end{align}
for $(t,x,y,z)\in[0,T]\times\mathcal{D}$ and final condition $Z^{\varphi,a}_s(s,x,y,z)=\varphi(s,x)$. Note that 
     \begin{align*}
        \partial_tZ^{\varphi,a}_s(t,x,y,z)=e^{r(s-t)}\left(-rU^{\varphi,a}_s(t,x,y,z)+\partial_t    U^{\varphi,a}_s(t,x,y,z) \right).
    \end{align*}
    and
    \begin{align*}
        \mathcal{L}^aZ^{\varphi,a}_s(t,x,y,z)=e^{r(s-t)}\mathcal{L}^aU^{\varphi,a}_s(t,x,y,z).
    \end{align*}
    Replacing that in \eqref{Apidbefore} we get that $U_s^{\varphi,a}$ satisfies the following PIDE
    \begin{align*}
        \partial_t U_s^{\varphi,a}(t,x,y,z)+\mathcal{L}^aU_s^{\varphi,a}(t,x,y,z)=rU_s^{\varphi,a}(t,x,y,z).
    \end{align*}
for $(t,x,y,z)\in[0,s]\times\mathcal{D}$ and final condition $U_s^{\varphi,a}(s,x,y,z)=\varphi(s,x)$. 
\end{proof}

\newpage
\printbibliography

\end{document}